\makeatletter
\def\input@path{{template/}}
\makeatother

\documentclass[conference]{IEEEtran}

\newif\ifpublish
\publishtrue

\usepackage{silence}
\usepackage{lmodern}
\usepackage[T1]{fontenc}
\usepackage{microtype}

\usepackage{amsmath}
\usepackage{amssymb}
\usepackage{amsthm}

\usepackage{graphicx}
\graphicspath{{data/plots/}{figures/}}
\usepackage{subcaption}
\usepackage{booktabs}

\usepackage{enumitem}
\setlist[itemize]{leftmargin=*,labelindent=0pt,itemsep=2pt,topsep=2pt}
\setlist[enumerate]{leftmargin=*,labelindent=0pt,itemsep=2pt,topsep=2pt}
\makeatletter
\renewcommand\paragraph{\@startsection{paragraph}{4}{\z@}
    {0.6\baselineskip \@plus 2\p@}{-0.5em}{\normalfont\normalsize\bfseries}}
\def\@IEEEsectpunct{\ }
\makeatother
\usepackage{xspace}
\usepackage{etoolbox}
\usepackage{xcolor}

\usepackage{algorithm}
\usepackage[noend]{algpseudocode}
\usepackage{float}

\usepackage{hyperref}
\hypersetup{
    colorlinks=true,
    linkcolor={magenta!80!black},
    citecolor={green!55!black},
    urlcolor={black!55},
    pdfborder={0 0 0},
}
\usepackage{cleveref}

\crefname{section}{Section}{Sections}
\Crefname{section}{Section}{Sections}
\crefname{subsection}{Section}{Sections}
\Crefname{subsection}{Section}{Sections}
\crefname{appendix}{Appendix}{Appendices}
\Crefname{appendix}{Appendix}{Appendices}
\crefname{algorithm}{Algorithm}{Algorithms}
\Crefname{algorithm}{Algorithm}{Algorithms}
\crefname{theorem}{Theorem}{Theorems}
\Crefname{theorem}{Theorem}{Theorems}
\crefname{lemma}{Lemma}{Lemmas}
\Crefname{lemma}{Lemma}{Lemmas}
\crefname{table}{Table}{Tables}
\Crefname{table}{Table}{Tables}
\crefname{equation}{Equation}{Equations}
\Crefname{equation}{Equation}{Equations}
\crefname{ALG@line}{line}{lines}
\Crefname{ALG@line}{Line}{Lines}

\newtheorem{theorem}{Theorem}
\newtheorem{lemma}[theorem]{Lemma}

\newcommand{\sysname}{\textsf{Orcaella}\xspace}
\newcommand{\orcdag}{\textsf{OrcDAG}\xspace}

\newcommand{\voteqc}{\textsf{VoteQC}\xspace}
\newcommand{\voteqcs}{\textsf{VoteQCs}\xspace}
\newcommand{\checkpointqc}{\textsf{CheckpointQC}\xspace}
\newcommand{\checkpointqcs}{\textsf{CheckpointQCs}\xspace}
\newcommand{\finalityqc}{\textsf{FinalityQC}\xspace}
\newcommand{\finalityqcs}{\textsf{FinalityQCs}\xspace}

\newcounter{claim}
\renewcommand{\theclaim}{\textbf{C\arabic{claim}}}
\newcommand{\claimitem}[1]{\item\refstepcounter{claim}\label{#1}\theclaim:}

\newcommand{\consensuscode}{
    \ifpublish
        \url{https://github.com/asonnino/mysticeti} (commit \texttt{96dee8d})
    \else
        \url{https://anonymous.4open.science/r/uncertified-dag-fabric}
    \fi
}
\newcommand{\validatorcode}{
    \ifpublish
        \url{https://github.com/asonnino/blockchain-validator} (commit \texttt{a7cc550
        })
    \else
        \url{https://anonymous.4open.science/r/blockchain-validator}
    \fi
}

\newcommand{\algsize}{\fontsize{7.5}{8.5}\selectfont}

\AtBeginEnvironment{algorithm}{\algsize}

\floatstyle{boxed}
\restylefloat{algorithm}
\restylefloat*{algorithm}

\begin{document}

\title{\sysname: Hybrid Fault Tolerance with Client-Selectable Finality Latency}

\ifpublish
    \author{
        \IEEEauthorblockN{Lefteris Kokoris Kogias, Alberto Sonnino}
        \IEEEauthorblockA{Mysten Labs\\ \texttt{\{lefteris, alberto\}@mystenlabs.com}}
    }
\else
    \author{}
    \pagestyle{plain}
    \thispagestyle{plain}
\fi
\maketitle

\begin{abstract}
    Classical partially synchronous state machine replication, as in
    PBFT~\cite{Castro1999PBFT}, tolerates $f$ Byzantine replicas among
    $n \geq 3f+1$ using three communication steps per request. Recent
    protocols such as
    Minimmit~\cite{Minimmit} achieve \emph{two-message-delay} decisions
    under stronger size assumptions, notably $n \geq 5f+1$ when any silent
    replica must be counted as a potential equivocator.
    Hydrangea~\cite{Hydrangea} and Kudzu~\cite{Kudzu} treat mixed Byzantine and crash faults, focusing on providing a fast-path under optimistic conditions while maintaining a fall-back commitment path similar to PBFT. In this paper, we also consider a mixed model, but focus on studying the fault tolerance of the 2-message-delay commit. For this, we prove a tight bound of
    $n \geq 5f+3c+1$. Extending this result, we also show that there exists a more resilient commit path that allows an extra $f_{abc} < n-3f-2c$ alive-but-corrupt~\cite{Malkhi2019FlexibleBFT} faults at 4-message-delays. Core liveness is claimed in executions with at most $f$ equivocators; if this regime is violated (e.g., AbC-induced forks), the protocol enters synchronous recovery, where only the resilient-path safety guarantee is preserved.
    As a result, for $f=16$, $c=6$, and $n=99$, we obtain a commit path that tolerates $22\%$ of replicas failing for liveness, $16\%$ equivocating for 1-RTT safety, and $54\%$ equivocating for 2-RTT safety, where one RTT equals two message delays for clients collocated with replicas.
\end{abstract}

\section{Introduction}
\label{sec:intro}

State machine replication (SMR) protocols are the foundational building blocks of modern decentralized systems~\cite{sok-consensus}. Under partial synchrony~\cite{Dwork1988PartialSync}, classical solutions like PBFT~\cite{Castro1999PBFT} tolerate up to $f$ Byzantine replicas among $n \geq 3f+1$, but require three all-to-all communication delays to reach consensus. As decentralized applications demand ever-lower latency, recent research has introduced \emph{two-message-delay} protocols (e.g., Minimmit~\cite{Minimmit}). However, achieving this optimal fast-path latency traditionally requires a much larger committee size of $n \geq 5f+1$, as every silent replica must be conservatively treated as a potential Byzantine equivocator.

\paragraph*{The challenge.}
To mitigate the severe $5f+1$ requirement, a natural approach is to distinguish between fully Byzantine faults (who may equivocate) and crash faults (who merely go silent). Recent work like Hydrangea~\cite{Hydrangea} and Kudzu~\cite{Kudzu} adopt this mixed fault model. However, existing hybrid systems focus on providing a fast-path under optimistic conditions, while eagerly falling back to a slower, PBFT-style 3-round path when faults increase.

This leaves an open challenge: \emph{How can we maximize the fault tolerance of the optimal 2-message-delay commit itself when relaxing the fault model?} Furthermore, if a client does not care about this lower latency, \emph{can we enhance the mixed fault model to provide even stronger, FlexibleBFT-style~\cite{Malkhi2019FlexibleBFT} safety guarantees?}

\paragraph*{Our solution.}
In this paper, we answer these questions by formally defining the tight quorum intersections required to maximize 2-message-delay commits under separate Byzantine ($f$) and crash-faulty ($c$) caps. We show that such a system necessitates:
\[
  n \;\geq\; 5f + 3c + 1
\]
with an optimal fast-path quorum of $q = n-f-c$, i.e., $4f+2c+1$ at the minimal committee size. At $c=0$, this recovers the familiar $5f+1$ regime~\cite{FaBPaxos,Minimmit}, but when crashes are separated, it allows for configurations with better liveness guarantees.

Building on this foundational result, we introduce \sysname, a hybrid protocol that exposes two explicit finality paths for clients. These paths are built using \emph{Quorum Certificates} (QCs)---cryptographic proofs consisting of $q$ matching signed messages from distinct replicas:
\begin{itemize}
  \item \textbf{Optimal Fast-Path (2-delay):} Clients can finalize quickly after collecting a single quorum of votes (\voteqc), providing safety against $f$ Byzantine faults.
  \item \textbf{Resilient Path (4-delay):} Clients willing to wait for two additional rounds (chaining a \checkpointqc and \finalityqc) gain resilience against an additional $f_{abc}$ \emph{alive-but-corrupt}~\cite{Malkhi2019FlexibleBFT} replicas. These are replicas that participate but may equivocate.
\end{itemize}

Crucially, the resilient path provides safety when the core execution assumptions are violated and forks occur. In executions with at most $f$ equivocators, the protocol retains its normal liveness behavior. If this regime is violated (e.g., AbC-induced forks), liveness is temporarily lost, replicas enter synchronous recovery, and resilient-path safety remains preserved throughout.

\paragraph*{Concrete trade-offs.}
This dual-path architecture allows operators to explicitly navigate the latency-safety-liveness trade-off. Depending on the expected environment, for a deployment of $n \approx 100$ replicas, an operator can tune their configuration (\Cref{fig:fault-coverage}):

\begin{figure}[t]
  \centering
  \includegraphics[width=\linewidth]{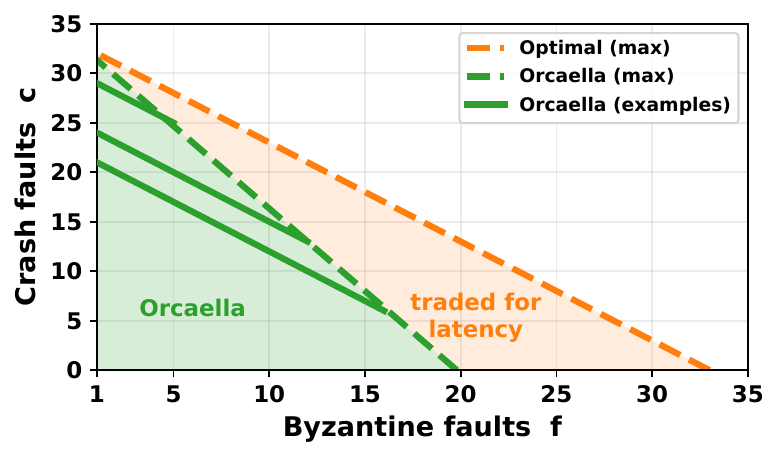}
  \caption{
    Fault-tolerance design space at $n=100$. Typical $3f+1$ protocols treat every fault as Byzantine, covering the region $f+c \le 33$ (orange) at the cost of requiring 3 message delays to commit. \sysname instead covers the region $5f+3c+1 \le n$ (green), trading some fault tolerance for its 2-message-delay commit. Each solid green line is one example config's runtime frontier: a cap $(f,c)$ tolerates any mix down to $(0,f+c)$.
  }
  \label{fig:fault-coverage}
\end{figure}

\begin{itemize}
  \item \textbf{Balanced ($f=12, c=13$):} Tolerates up to 25 offline replicas for liveness, ensures 1-RTT (two-message-delay) safety against 12 Byzantine faults, and 2-RTT safety against a maximum of 49 total equivocators (strictly below the intersection bound $T = 50$).
  \item \textbf{Byzantine-heavy ($f=16, c=6$):} Optimizing for more malicious conditions slightly reduces liveness tolerance (22 offline replicas) but boosts 1-RTT safety to 16 and 2-RTT safety to a maximum of 54 total equivocators ($< T = 55$).
  \item \textbf{Crash-heavy ($f=5, c=25$):} Assuming a more benign but flaky network maximizes liveness (tolerating 30 offline replicas) while still offering 2-RTT safety against a maximum of 40 total equivocators ($< T = 41$).
\end{itemize}
This flexibility allows deployments to adapt to their exact threat models without sacrificing the optimal 2-message-delay commit. Although for the scope of this paper and all our results we consider $f>0$ we note that \sysname can actually be configured with $f=0$ and $n=2c+1$ to provide a CFT variant without any code changes other than configuring the quorums. \Cref{sec:crash-only} shows an evaluation of this variant for completeness.

\paragraph*{Contributions.}
We make the following contributions:
\begin{enumerate}
  \item The tight quorum inequalities ($n \geq 5f+3c+1$) that maximize the fault tolerance of 2-delay BFT consensus (\cref{thm:main}).
  \item \sysname, a protocol that exposes the 2-delay and 4-delay paths, including a fully specified view change and fork recovery mechanism.
  \item FlexibleBFT-style client analyses formalizing the safety vs.\ liveness trade-offs and proving the 4-delay path remains safe against an additional $f_{abc} < n - 3f - 2c$ alive-but-corrupt faults.
  \item \orcdag, a DAG-based instantiation of \sysname over an uncertified-DAG fabric, and its evaluation on realistic deployment quantifying the latency gained by trading fault tolerance.
\end{enumerate}

\section{Network and Fault Model}
\label{sec:model}

We use standard message-passing state machine replication~\cite{Schneider1990SMR} among $n$
\emph{replicas} (the \emph{leader} or \emph{primary} proposes each
slot) under partial synchrony~\cite{Dwork1988PartialSync}.
Channels are authenticated; digital signatures identify senders. We distinguish four replica \emph{behaviors}:
\begin{itemize}
  \item \emph{Byzantine} replicas: may crash or send arbitrary messages including \emph{equivocations} (inconsistent signed statements to different recipients) on any protocol message;
  \item \emph{Crash-faulty} replicas: follow the protocol while active and may then permanently stop sending messages without equivocating;
  \item \emph{Alive-but-Corrupt} (AbC) replicas~\cite{Malkhi2019FlexibleBFT}: remain live and participate in the protocol but may equivocate on any message they send;
  \item \emph{Correct} replicas: follow the protocol and remain live throughout.
\end{itemize}
Rather than fixing a single static fault partition, \sysname states each guarantee over one of two \emph{nested} families of executions:
\begin{itemize}
  \item \emph{Core executions}, in which at most $f$ replicas equivocate (Byzantine) and at most $c$ crash, leaving at least $n-f-c$ correct replicas. Every core execution enjoys Fast-Path safety \emph{and} liveness; all liveness analysis in this paper assumes this family, i.e., at most $f$ equivocators.
  \item \emph{Extended executions}, in which at most $f+f_{abc}$ replicas equivocate ($f$ Byzantine plus $f_{abc} < n-3f-2c$ AbC) and at most $c$ crash, leaving at least $n-f-c-f_{abc}$ correct replicas. In this strictly larger family only Resilient-Path safety is guaranteed; liveness may be temporarily lost until synchronous recovery (\cref{alg:fork_recovery}).
\end{itemize}
The same protocol runs unmodified in both families; each client chooses which finality rule to rely on---\voteqc (Fast Path) or \finalityqc (Resilient Path)---and thereby which family's guarantee it obtains.

The standard partial synchrony network model~\cite{Dwork1988PartialSync} assumes
that after an unknown Global Stabilization Time (GST), message delays between
correct replicas are bounded by a known constant~$\Delta$. The adversary controls all non-honest replicas up to the misbehavior they are allowed to perform and the network
scheduling (subject to this bound after GST). It cannot break cryptography.

\begin{table}[t]
  \caption{Notation}
  \label{tab:notation}
  \centering\footnotesize
  \begin{tabular}{@{}cl@{}}
    \toprule
    Symbol    & Meaning                                                                 \\ \midrule
    $n$       & Total replicas                                                          \\
    $f$       & Byzantine (equivocating) cap                                            \\
    $c$       & Crash-faulty cap                                                        \\
    $f_{abc}$ & Alive-but-Corrupt cap (Resilient Path client analysis)                  \\
    $q$       & Quorum threshold, $q = n-f-c$                                           
    \\
    $k$       & View-Change accept threshold, $k = 2f+c+1$                              \\
    $\sigma$  & State root after applying a slot ($\sigma_s = H(\sigma_{s-1} \Vert h)$) \\
    $\Delta$  & Post-GST bound on message delay                                         \\
    \bottomrule
  \end{tabular}
\end{table}

\section{The \sysname Protocol}
\label{sec:protocol}

This section instantiates the fault model and Fast-Path/Resilient Path split from \Cref{sec:model} into a message-level protocol.
It uses the thresholds derived in \cref{thm:main}: $n \geq 5f+3c+1$ and $q = n-f-c$ (equal to $4f+2c+1$ at the minimal $n$).
Replicas communicate on authenticated channels under partial synchrony. The protocol proceeds in views $v \in \mathbb{N}$, each with a designated leader $\textsc{Leader}(v)$, and assigns proposals to sequence numbers $s \in \mathbb{N}$.

\subsection{The protocol}

\cref{alg:protocol} specifies the protocol, covering both normal operation and leader recovery. Normal operation is divided into a Fast-Path (Steps~0--1, \crefrange{step:0}{step:1}) that guarantees safety against $f$ Byzantine faults in two message delays, and a Resilient Path (Steps~2--4, \crefrange{step:2}{step:4}) that provides extended safety against alive-but-corrupt (AbC) faults at the cost of two additional message delays.

\begin{algorithm*}[t]
  \caption{\sysname: Normal Operation \& Leader Recovery for replica $r$}
  \label{alg:protocol}
  \algsize
  \begin{algorithmic}[1]
    \State \textbf{Definitions \& State Variables:}
    \Statex \indent \textbf{Thresholds}: $q = n-f-c$ and $k = 2f+c+1$.
    \Statex \indent \voteqc: $q$ matching \textsc{Vote} messages for $(v,s,h)$ with \textsf{accept}.
    \Statex \indent \checkpointqc: $q$ matching \textsc{ChkProp} messages for $(s,h,\sigma)$.
    \Statex \indent \finalityqc: $q$ matching \textsc{ChkWitness} messages for $(s,h,\sigma)$.
    \Statex \indent $\textsf{current\_view} \gets 0$
    \Statex \indent $\mathcal{O} \gets \emptyset$ \Comment{Branch adopted from the latest \textsc{NewView} ($\bot$ = unconstrained)}
    \Statex
    \State \label{step:0} \textbf{(step 0)} \textbf{upon} $r = \textsc{Leader}(v)$ \textbf{and} $v = \textsf{current\_view}$ \textbf{and} new payload $B_{v,s}$ \textbf{do}
    \State \indent \textsf{broadcast} $\langle \textsc{Propose}, v, s, B_{v,s} \rangle_r$
    \Statex
    \State \label{step:1} \textbf{(step 1)} \textbf{upon} receiving $\langle \textsc{Propose}, v, s, B \rangle_L$ from $L = \textsc{Leader}(v)$ \textbf{do}
    \State \indent \textbf{if} $v = \textsf{current\_view}$ \textbf{and} $r$ has not yet voted for slot $s$ in view $v$ \textbf{then}
    \State \indent\indent \textbf{if} $\mathcal{O}[s] \in \{\bot, H(B)\}$ \textbf{then} \Comment{Vote only consistent with the adopted branch}
    \State \indent\indent\indent \textsf{broadcast} $\langle \textsc{Vote}, v, s, H(B), \textsf{accept} \rangle_r$
    \Statex
    \State \label{step:2} \textbf{(step 2)} \textbf{upon} observing a valid \voteqc $Q$ for $(v,s,h)$ \textbf{do}
    \State \indent \textbf{if} no \textsc{ChkProp} has been broadcast for slot $s$ \textbf{then}
    \State \indent\indent $\sigma_{s} \gets \textsc{Apply}(h)$ \Comment{Deterministic state transition}
    \State \indent\indent \textsf{broadcast} $\langle \textsc{ChkProp}, s, h, \sigma_{s}, Q \rangle_r$ \Comment{$Q$ justifies the proposal; excluded from quorum matching}
    \Statex
    \State \label{step:3} \textbf{(step 3)} \textbf{upon} observing a \checkpointqc $C$ for $(s,h,\sigma)$ \textbf{where} $\sigma = \textsc{Apply}(h)$ \textbf{do}
    \State \indent \textsf{broadcast} $\langle \textsc{ChkWitness}, s, h, \sigma \rangle_r$
    \Statex
    \State \label{step:4} \textbf{(step 4)} \textbf{upon} collecting $q$ matching $\langle \textsc{ChkWitness}, s, h, \sigma \rangle$ messages \textbf{do}
    \State \indent \textsf{output} $\langle \textsc{\finalityqc}, s, h, \sigma \rangle$ \Comment{Resilient Path finality decision for slot $s$}
    \Statex
    \State \textbf{upon} timeout in view $v$ \textbf{or} timeout awaiting $\langle \textsc{NewView}, v \rangle$ \textbf{do} \Comment{Timeouts double per attempt}
    \State \indent $\textsf{current\_view} \gets \bot$ \Comment{Halt normal operation}
    \State \indent $\mathcal{P}^r \gets \{ \text{highest-view } \langle \textsc{Vote}, v', s, h, \textsf{accept} \rangle \text{ cast by } r \text{ for each } s \}$
    \State \indent \textsf{broadcast} $\langle \textsc{ViewChange}, v+1, \mathcal{P}^r \rangle_r$
    \Statex
    \State \textbf{upon} $r = \textsc{Leader}(v)$ and receiving $q$ valid $\langle \textsc{ViewChange}, v, \mathcal{P}_i \rangle$ messages $\mathcal{V}$ \textbf{do}
    \State \indent \textbf{for each} slot $s$ represented in $\mathcal{V}$ \textbf{do}
    \State \indent\indent $\mathcal{K} \gets \{ h \mid h \text{ is the highest-view vote for } s \text{ in at least } k \text{ reports in } \mathcal{V} \}$ \Comment{Aggregated per digest, across views}
    \State \indent\indent \textbf{if} $\mathcal{K} \neq \emptyset$ \textbf{then}
    \State \indent\indent\indent $\mathcal{O}[s] \gets$ digest in $\mathcal{K}$ with the highest reported view \Comment{Singleton if slot $s$ has a \voteqc (\cref{lem:view-change-invariance}); ties broken by smallest digest}
    \State \indent\indent \textbf{else}
    \State \indent\indent\indent $\mathcal{O}[s] \gets \bot$ \Comment{No locked value; safe to propose arbitrary valid block}
    \State \indent \textsf{broadcast} $\langle \textsc{NewView}, v, \mathcal{V}, \mathcal{O} \rangle_r$
    \Statex
    \State \textbf{upon} receiving valid $\langle \textsc{NewView}, v, \mathcal{V}, \mathcal{O} \rangle_L$ from $L = \textsc{Leader}(v)$ \textbf{do}
    \State \indent verify $\mathcal{O}$ against $\mathcal{V}$ \Comment{Deterministically recompute $\mathcal{K}$ and $\mathcal{O}$}
    \State \indent $\textsf{current\_view} \gets v$ \Comment{Enter new view}
    \State \indent adopt branch $\mathcal{O}$ and resume normal operation
  \end{algorithmic}
\end{algorithm*}

\paragraph*{Client finalization paths.}
Clients finalize transactions by observing the certificates produced by the protocol. The latency bounds stated in this paper (e.g., two or four message delays) assume that clients are collocated on the replicas participating in the consensus protocol. If clients are external, their perceived finality latency naturally increases by the round-trip time (RTT) to communicate with the replica set, as is standard in all BFT protocols.
\begin{itemize}
  \item \textbf{Fast-Path (2-delay):} A client finalizes $(v,s,h)$ immediately upon collecting a \voteqc. This provides optimal latency and safety against $f$ Byzantine faults.
  \item \textbf{Resilient Path (4-delay):} A client finalizes $(s,h,\sigma)$ upon obtaining a \finalityqc (output by a replica in Step~4, \cref{step:4}) together with its underlying \checkpointqc. This trades higher latency for extended safety against an additional $f_{abc}$ alive-but-corrupt faults, without interfering with the core protocol's liveness~\cite{Malkhi2019FlexibleBFT}.
\end{itemize}

\paragraph*{Protocol description.}
The leader of view $v$ proposes a batch of transactions. Replicas validate the proposal and cast an \textsf{accept} vote, but only if the proposal is consistent with the branch $\mathcal{O}$ adopted from the latest \textsc{NewView} (\cref{step:1}); this guard is what carries Fast-Path commits across view changes (\cref{lem:view-change-invariance}). Correct replicas send at most one vote per slot $(v,s)$. If a client gathers $q$ such \textsf{accept} votes, it forms a \voteqc and can immediately safely execute the payload (Fast-Path). If the leader is faulty~\cite{be-aware-leaders} or the network is asynchronous, replicas eventually timeout and broadcast \textsc{ViewChange} messages, which collectively serve as an implicit abandon for uncommitted slots.

To support the Resilient Path, replicas that observe a \voteqc compute the resulting deterministic state root $\sigma$ and broadcast a checkpoint proposal. Each \textsc{ChkProp} carries its justifying \voteqc, so conflicting checkpoint proposals surface conflicting \voteqcs as fork evidence. We assume state roots bind the execution history: $\sigma_s$ commits to the previous root and the applied digest (e.g., $\sigma_s = H(\sigma_{s-1} \,\Vert\, h)$), so a root at height $s$ determines the checkpointed content at every height below it. \textbf{Safety is strictly guarded by Step~2} (\cref{step:2}): to prevent conflicting checkpoints, an honest replica broadcasts at most one \textsc{ChkProp} per slot height. For this reason, checkpoint artifacts deliberately carry no view number: a slot's \voteqc may form in view $v$ at one replica and in a later view $v'$ at another (with the same digest, by Fast-Path safety), and view-tagged \textsc{ChkProp}s would split each replica's single permitted proposal across views, permanently blocking the checkpoint. While this strict locking rule means \emph{liveness} can be lost if honest replicas split their \textsc{ChkProp} messages across different branches (a trade-off we explicitly accept), it provides an ironclad safety guarantee. Once $q$ replicas agree on this checkpoint (forming a \checkpointqc), they lock it by broadcasting a \textsc{ChkWitness} message. Note that \textsc{ChkWitness} messages attest to the checkpoint \emph{content} $(s,h,\sigma)$ rather than to one specific certificate: two \checkpointqcs over the same content but different signer sets may circulate, and binding witnesses to a particular certificate would needlessly split them. To safely support a checkpoint, an honest replica will broadcast a \textsc{ChkWitness} as long as the checkpoint's state $\sigma$ matches its own deterministic execution of the payload ($\sigma = \textsc{Apply}(h)$), even if its single permitted \textsc{ChkProp} for that height was already spent on a different, unfinalized branch. Once $q$ such witness messages are collected, a replica outputs a \finalityqc (Step~4, \cref{step:4}) as its Resilient Path finality decision for that slot.

Notice that with $f_{abc}$ alive-but-corrupt faults, there can be \voteqc forks across different views. However, because honest replicas only broadcast a single \textsc{ChkProp} per height, the double-signing bound (\cref{lem:unique-checkpointqc}) guarantees that at most one valid \checkpointqc can ever be formed per height. Any resulting loss of liveness on the Resilient Path is cleanly resolved by the Fork Recovery protocol (\cref{alg:fork_recovery}) during a synchronous epoch. Steps~2--4 (\crefrange{step:2}{step:4}) ensure that this single \checkpointqc is witnessed by the network and the resulting \finalityqc is propagated before resilient finalization. Notice that every protocol artifact (\voteqc, \checkpointqc, \finalityqc) uses the exact same quorum threshold $q$.

\paragraph*{View-Change details.}
The view change protocol is structurally similar to both PBFT and modern protocols like Streamlet~\cite{Chan2020Streamlet} or HotStuff~\cite{hotstuff,jolteon}. If $\textsc{Leader}(v+1)$ fails to assemble a valid \textsc{NewView} before a timeout, replicas re-issue \textsc{ViewChange} messages for view $v+2$ with doubled timeouts, following standard view synchronization~\cite{Castro1999PBFT}; a replica that observes a valid \textsc{ViewChange} or \textsc{NewView} for a view higher than its own joins it immediately, broadcasting its own \textsc{ViewChange} for that view. Assuming the core limit of $f$ Byzantine faults holds (i.e., there are no \voteqc forks), any two sets of $q$ \textsc{ViewChange} messages intersect at an honest replica ($2q - n > f$). This guarantees that any \voteqc used by a client to finalize a Fast-Path commit is reported to the new leader and preserved in the new view's starting state $\mathcal{O}$. The new leader may hold only the digest $h$ for a re-proposed slot; it fetches the block body from the $\geq q-f$ honest replicas that voted for $h$, since honest replicas vote only on proposals they store. Since replicas may re-vote for the same digest across views, the leader counts the reported votes per digest rather than per view (\cref{lem:view-change-invariance}). By operating over a large $5f+1$ style quorum, the protocol acts as a generalized 2-chain Streamlet, enabling rapid 2-message-delay finality instead of the standard 3-chain required under $3f+1$. If the $f$ limit is exceeded (e.g., by AbC faults double-voting) and a \voteqc fork occurs, Fast-Path safety is broken and preserving it is no longer required; instead, replicas halt and rely on Fork Recovery to salvage the Resilient Path. Note that for simplicity, \cref{alg:protocol} requires replicas to send all their highest-view votes. In a practical implementation, replicas would use \checkpointqcs to prove a globally agreed upon state and safely truncate the required vote history, significantly reducing the size of \textsc{ViewChange} messages~\cite{beluga}.

\subsection{Fork recovery}
\label{sec:proto-fork}

\cref{alg:fork_recovery} specifies the procedure for detecting and recovering from forks (when the core $f$ Byzantine limit is exceeded but clients rely on Resilient Path safety).

\paragraph*{Synchronous network model.}
Unlike normal operation which relies on partial synchrony, this recovery procedure operates under a strictly synchronous network model. We assume that message delays between correct replicas are guaranteed to be bounded by a known, pessimistic constant $\Delta_{\mathrm{sync}}$. Replicas anchor the epoch's lockstep round structure at the delivery of the first valid \textsc{Alarm} (which every correct replica re-broadcasts), so the round boundaries of correct replicas are offset by at most $\Delta_{\mathrm{sync}}$---a bounded skew that synchronous broadcast protocols tolerate by standard techniques~\cite{Abraham2020Sync}.

\begin{algorithm*}[t]
  \caption{Fork Detection and Repair for replica $r$}
  \label{alg:fork_recovery}
  \algsize
  \begin{algorithmic}[1]
    \State \textbf{upon} observing conflicting \voteqcs or a valid $\langle \textsc{Alarm}, \text{evidence} \rangle$ \textbf{do}
    \State \indent halt normal operation
    \State \indent \textsf{broadcast} $\langle \textsc{Alarm}, \text{evidence} \rangle_r$
    \State \indent initiate a Byzantine Broadcast protocol (e.g., Dolev-Strong)
    \State \indent \textsf{to reliably exchange all locally known \checkpointqcs and \voteqcs}
    \Statex
    \State \textbf{upon} broadcast protocol completes \textbf{do}
    \State \indent $\mathcal{C} \gets \text{all \checkpointqcs delivered by the broadcast}$
    \State \indent \textbf{if} $\mathcal{C} \neq \emptyset$ \textbf{then}
    \State \indent\indent $C_{max} \gets \text{the } C \in \mathcal{C} \text{ with the highest slot}$ \Comment{Mathematically guaranteed unique per height}
    \State \indent\indent $\mathcal{O}_{canon} \gets \text{branch of } C_{max} \text{ extended per slot with the lowest-digest delivered \voteqc, up to the first gap}$
    \State \indent \textbf{else}
    \State \indent\indent $\mathcal{O}_{canon} \gets \text{the branch built by the same deterministic rule from the initial state}$
    \State \indent start a new epoch: adopt the final state of $\mathcal{O}_{canon}$ as the new genesis
    \State \indent resume normal operation from the new genesis
  \end{algorithmic}
\end{algorithm*}

\paragraph*{Fork recovery details.}
While the core protocol naturally resists up to $f$ fully Byzantine faults, Resilient Path clients may additionally assume $f_{abc}$ alive-but-corrupt replicas. These replicas behave correctly during normal operation but might double-sign conflicting \textsc{ChkProp} messages to intentionally cause forks.

However, because honest replicas only broadcast a single \textsc{ChkProp} per height, the double-signing bound (\cref{lem:unique-checkpointqc}) explicitly prevents the $f_{abc}$ alive-but-corrupt replicas from successfully forming two conflicting \checkpointqcs. Thus, there can be at most one valid \checkpointqc per height.

If a \voteqc fork occurs (exceeding the core $f$ Byzantine limit), correct replicas immediately detect it via cryptographic evidence, halt normal operation, and force a global repair by broadcasting an \textsc{Alarm}. Valid \textsc{Alarm} evidence is a pair of \voteqcs for the same slot with different digests; since every \textsc{ChkProp} carries its justifying \voteqc, a checkpoint split at any height surfaces such a pair. False alarms are impossible in core executions, where conflicting \voteqcs would require more than $f$ double-signers. Upon entering this synchronous recovery epoch, every replica initiates its own Byzantine Broadcast instance (such as Dolev-Strong~\cite{DolevStrong1983}), $n$ in parallel, to reliably exchange all locally known \checkpointqcs and \voteqcs; correct replicas then operate on the union of the delivered sets. Because Dolev-Strong does not rely on an honest majority, it guarantees that all correct replicas will output the exact same set of valid certificates, even if the total number of faulty replicas ($f + f_{abc}$) constitutes a dishonest majority ($>n/2$). However, if the deployment guarantees that $f + f_{abc}$ remains a strict minority ($<n/2$), this recovery step can be significantly optimized by deploying an honest-majority synchronous broadcast protocol that terminates much faster than Dolev-Strong's $(f+f_{abc})+1$ rounds (e.g., expected constant-round protocols~\cite{Katz2006Expected} or practical synchronous SMR implementations like Sync HotStuff~\cite{Abraham2020Sync}). When the broadcast protocol completes, replicas collect all delivered \checkpointqcs across all slots. At most one \checkpointqc exists per height (\cref{lem:unique-checkpointqc}), and every \checkpointqc stored by a correct replica is chain-consistent with all higher ones (\cref{lem:checkpoint-chain}); the highest delivered \checkpointqc therefore extends every checkpoint a correct replica holds---in particular every client-finalized one. Replicas simply identify the highest \checkpointqc and adopt its branch as canonical, extending it deterministically (per slot, the lowest-digest delivered \voteqc, up to the first gap), so all correct replicas derive the same genesis; nothing above $C_{max}$ is resilient-finalized, so any deterministic choice is safe. If no \checkpointqc was ever formed, the same rule applies from the initial state. This guarantees that Resilient Path safety remains unbroken.

\paragraph*{Post-recovery epochs.}
Recovery concludes by opening a new epoch: the final state of $\mathcal{O}_{canon}$ becomes the agreed genesis, and normal operation restarts with fresh views, slots, and checkpoint heights under an incremented epoch number. All artifacts carry the epoch number; certificates are ordered lexicographically by (epoch, height), and the single-\textsc{ChkProp} lock applies per (epoch, height). Correct replicas never sign artifacts of a closed epoch, so no stale checkpoint of the old epoch can gather the $q$ witnesses required for a \finalityqc. Each epoch is thus an independent execution of \sysname whose genesis commits to all previously finalized state, so \cref{thm:fast-path-safety,thm:liveness} apply per epoch and Resilient Path safety spans epochs (\cref{lem:checkpoint-chain,thm:resilient-safety}). We acknowledge that recovery restores liveness only until the next fork: the replicas that caused it remain in the committee and can fork the new epoch again. Since every fork exposes more than $f$ provable double-signers, we leave to the operator the policy for seeding the new epoch's committee, e.g., excluding the provably misbehaving parties, as well as the incentive design that further protects clients~\cite{cobra}.

\section{Safety and Liveness Proofs}
\label{sec:proofs}

We show that \sysname satisfies safety and liveness under the network and fault model defined in \Cref{sec:model}. Throughout, the quorum threshold is $q = n-f-c$ and the view-change accept threshold is $k = 2f+c+1$; at the minimal committee size $n = 5f+3c+1$, $q = 4f+2c+1$. We derive both thresholds directly through these proofs, ensuring a rigorous foundation for the optimal $n \ge 5f+3c+1$ bound.

\begin{theorem}[Tight Bound Necessity]
    \label{thm:main}
    For \sysname (\cref{alg:protocol}) to be live under $f$ Byzantine and $c$ crash faults while keeping Fast-Path commits safe across view changes, $n \ge 5f+3c+1$ is necessary. At $n = 5f+3c+1$, the thresholds are forced: $q = 4f+2c+1 = n-f-c$ and $k = 2f+c+1$.
\end{theorem}

\begin{proof}
    We derive three counting constraints; each is necessary against an explicit adversarial strategy. Throughout, $\mathcal{V}$ denotes the $q$ \textsc{ViewChange} messages collected by the new leader (\cref{alg:protocol}).

    \emph{Liveness (L): $q \le n-f-c$.}
    Byzantine replicas may remain silent and crash-faulty replicas may stop at any time, so only $n-f-c$ replicas are guaranteed to respond. Both \voteqc formation and the collection of $q$ \textsc{ViewChange} messages must complete without faulty participation, hence $q \le n-f-c$.

    \emph{Reachability (R): $k \le 2q-n-f$.}
    Suppose a client finalizes $h$ for slot $s$ via a \voteqc with vote quorum $Q_{vote}$, $|Q_{vote}| = q$. The adversary crashes the (up to $c$) crash-faulty members of $Q_{vote}$ before the view change and schedules delivery so that $\mathcal{V}$ contains as few members of $Q_{vote}$ as possible. Since $|\mathcal{V}| = q$ and only $n-q$ replicas lie outside $Q_{vote}$, $\mathcal{V}$ contains at least $2q-n$ members of $Q_{vote}$, of which up to $f$ are Byzantine and may misreport; only $2q-n-f$ reports for $h$ are thus guaranteed. If $k > 2q-n-f$, the adversary makes slot $s$ miss the threshold, so even an honest new leader sets $\mathcal{O}[s] = \bot$ and proposes a fresh block $B'$ with $H(B') \ne h$. Non-equivocating replicas, who have not yet voted in the new view, vote for $B'$, forming a conflicting \voteqc and violating safety. Hence $k \le 2q-n-f$.

    \emph{Exclusivity (E): $k \ge n-q+f+1$.}
    While a \voteqc for $h$ exists, at least $q-f$ non-equivocating replicas voted for $h$. Since each \textsc{ViewChange} message carries a single highest-view vote per slot, each non-equivocating replica supports at most one digest for $s$: at most $(n-f)-(q-f) = n-q$ non-equivocating replicas can report a conflicting digest $h'$, joined by up to $f$ Byzantine misreports, so any $\mathcal{V}$ contains at most $n-q+f$ reports for $h'$. This bound is realizable: a Byzantine leader of view $v$ equivocates between $h$ and $h'$; $q-f$ non-equivocating replicas vote for $h$ (completing the \voteqc together with the $f$ Byzantine), the remaining $n-q$ vote for $h'$, and the Byzantine replicas report $h'$ during the view change. By (R) the digest $h$ always reaches $k$ supporting reports; if additionally $k \le n-q+f$, the conflicting digest $h'$ also reaches $k$, and the Byzantine replicas fabricate a report view higher than any honest re-vote for $h$, so the highest-view adoption rule selects $h'$ even under an honest leader, overwriting the finalized $h$. Hence $k \ge n-q+f+1$.

    \emph{Combining.}
    (R) and (E) give $n-q+f+1 \le k \le 2q-n-f$, i.e., $3q \ge 2n+2f+1$. Substituting (L), $3(n-f-c) \ge 3q \ge 2n+2f+1$, hence $n \ge 5f+3c+1$.

    \emph{Forced thresholds.}
    At $n = 5f+3c+1$, the bound $3q \ge 2n+2f+1 = 12f+6c+3$ gives $q \ge 4f+2c+1$, while (L) gives $q \le n-f-c = 4f+2c+1$; thus $q = 4f+2c+1$. Then (E) yields $k \ge n-q+f+1 = 2f+c+1$ and (R) yields $k \le 2q-n-f = 2f+c+1$; thus $k = 2f+c+1$. Sufficiency of these thresholds is established by the remaining lemmas in this section.
\end{proof}

\subsection{Quorum intersection and Fast-Path safety}

\begin{lemma}[Quorum Intersection]
    \label{lem:intersection}
    Let $n \ge 5f+3c+1$ and $q = n-f-c$. The intersection of any two quorums of size $q$ contains at least $2f+c+1$ non-equivocating replicas.
\end{lemma}
\begin{proof}
    The intersection of two sets of size $q$ in a universe of size $n$ has size at least
    $2q - n = 2(n-f-c) - n = n-2f-2c$.
    Under the client-side assumption for the Fast-Path (\Cref{sec:model}), there are at most $f$ equivocating replicas in the system. Thus, the number of non-equivocating replicas in the intersection (which includes correct, crash-faulty, and AbC replicas: in core executions, AbC replicas do not equivocate) is at least $(n-2f-2c) - f = n-3f-2c \ge 2f+c+1$, where the last inequality uses $n \ge 5f+3c+1$.
\end{proof}

\begin{lemma}[Fast-Path Uniqueness]
    \label{lem:fast-path-uniqueness}
    For any given view $v$ and slot $s$, there can be at most one valid \voteqc.
\end{lemma}
\begin{proof}
    A \voteqc requires $q$ \textsf{accept} votes from distinct replicas for a specific digest $h$. Suppose, for the sake of contradiction, that two \voteqcs exist in view $v$ for slot $s$: one for $h$ and one for $h'$ (with $h \neq h'$). By \cref{lem:intersection}, their respective quorums intersect in at least $2f+c+1 \ge f+1$ non-equivocating replicas. This implies that at least one non-equivocating replica cast an \textsf{accept} vote for both $h$ and $h'$ in the same view $v$, which strictly violates the definition of a non-equivocating replica. Thus, $h = h'$.
\end{proof}

\subsection{View-Change safety}

To guarantee that a finalized Fast-Path commit is never overwritten by a subsequent leader, the View-Change protocol must enforce that any newly proposed branch $\mathcal{O}$ preserves previously committed digests. This is achieved via the view-change accept threshold $k = 2f+c+1$.

\begin{lemma}[View-Change Invariance]
    \label{lem:view-change-invariance}
    If a valid \voteqc for $(v, s, h)$ is formed, then for any subsequent view $v' > v$, if a non-equivocating replica adopts a branch $\mathcal{O}$ proposed by $\textsc{Leader}(v')$, it must be that $\mathcal{O}[s] = h$.
\end{lemma}
\begin{proof}
    We proceed by strong induction on $v'$, with the following induction hypothesis (IH): for every view $w$ with $v < w < v'$, any branch adopted by a non-equivocating replica in view $w$ satisfies $\mathcal{O}[s] = h$. (For $v' = v+1$ this range is empty and the argument below applies unchanged.)

    Let $Q_{vote}$ be the quorum of $q$ \textsf{accept} votes forming the \voteqc for $(v, s, h)$, and let $\mathcal{V}$ be the set of $q$ \textsc{ViewChange} messages collected by $\textsc{Leader}(v')$. By \cref{lem:intersection}, the signers of $Q_{vote}$ and of $\mathcal{V}$ intersect in at least $k$ non-equivocating replicas.

    We first claim that the highest-view vote for slot $s$ reported by any non-equivocating member of $Q_{vote}$ certifies $h$. Such a replica cast exactly one vote for $s$ in view $v$, namely for $h$. Consider any view $w$ with $v < w < v'$ in which it also voted for $s$: by the consistency guard of \cref{step:1}, it voted only for a proposal matching the branch it adopted in view $w$, and by the IH that branch satisfies $\mathcal{O}[s] = h$; hence that vote is also for $h$. Votes in views below $v$ are dominated by the view-$v$ vote. Its highest-view report for $s$ is therefore $(w', h)$ for some $w' \ge v$, proving the claim. In particular, the at least $k$ non-equivocating members of $Q_{vote}$ whose reports appear in $\mathcal{V}$ all support $h$, so $h \in \mathcal{K}$.

    We now show no conflicting digest enters $\mathcal{K}$. Each \textsc{ViewChange} message carries a single highest-view vote for $s$, so each replica supports at most one digest. By the claim, every non-equivocating member of $Q_{vote}$ supports $h$; hence at most $(n-f)-(q-f) = n-q$ non-equivocating replicas can support a digest $h' \ne h$, joined by at most $f$ Byzantine (possibly fabricated) reports. In total, $h'$ gathers at most $n-q+f = 2f+c < k$ reports, so $h' \notin \mathcal{K}$.

    Thus $\mathcal{K} = \{h\}$ and the view-change rule forces $\mathcal{O}[s] = h$. Since replicas deterministically recompute $\mathcal{K}$ from $\mathcal{V}$ before adopting a branch, a \textsc{NewView} carrying any other value for $\mathcal{O}[s]$ fails verification and is rejected; hence any branch adopted by a non-equivocating replica in view $v'$ satisfies $\mathcal{O}[s] = h$.
\end{proof}

\begin{theorem}[Fast-Path Safety]
    \label{thm:fast-path-safety}
    If two clients finalize $(v, s, h)$ and $(v', s, h')$ via the Fast-Path, then $h = h'$.
\end{theorem}
\begin{proof}
    If $v = v'$, \cref{lem:fast-path-uniqueness} guarantees $h = h'$. If $v < v'$, the first finalization implies a \voteqc was formed in view $v$. By \cref{lem:view-change-invariance}, every branch $\mathcal{O}$ adopted by a non-equivocating replica in view $v'$ satisfies $\mathcal{O}[s] = h$, and by the consistency guard of \cref{step:1}, such replicas vote in view $v'$ only for a proposal $B'$ with $H(B') = \mathcal{O}[s] = h$. A \voteqc for slot $s$ in view $v'$ contains at least $q - f \geq 1$ votes from non-equivocating replicas; hence it certifies $h$. Thus, $h = h'$.
\end{proof}

\subsection{Liveness}

\begin{theorem}[Liveness (Core Execution)]
    \label{thm:liveness}
    Under partial synchrony, in executions where at most $f$ replicas equivocate (i.e., no AbC-induced fork epoch), the protocol ensures that correct replicas eventually commit new proposals; Resilient Path certificates are then produced as part of normal progress.
\end{theorem}
\begin{proof}
    After the Global Stabilization Time (GST), message delays are bounded by $\Delta$. By standard view-synchronization arguments~\cite{Dwork1988PartialSync}, correct replicas will eventually enter a common view $v$ with a correct leader. During this view, the correct leader will broadcast a valid \textsc{NewView} message followed by \textsc{Propose} messages. Because there are at least $n - f - c = q$ correct replicas that remain active, they will all receive the leader's proposals, validate them, and broadcast \textsc{Vote} messages within the timeout. The leader (and all clients) will collect these $q$ votes to form a \voteqc, achieving Fast-Path finality. Consequently, every honest replica that observes a \voteqc for slot $s$ broadcasts a \textsc{ChkProp}, and these proposals all carry identical content: any \voteqc for slot $s$---in any view---certifies the same digest $h$ (\cref{lem:fast-path-uniqueness,lem:view-change-invariance}), honest replicas execute slots in order from the same checkpointed prefix, and $\textsc{Apply}$ is deterministic, so every honest \textsc{ChkProp} for height $s$ equals $(s, h, \sigma_s)$. Since checkpoint artifacts carry no view tag, the $n-f-c \ge q$ matching proposals aggregate into a \checkpointqc regardless of which view's \voteqc each replica observed first. The $\sigma = \textsc{Apply}(h)$ check then passes at every honest replica, so $q$ matching \textsc{ChkWitness} messages are collected and a \finalityqc is output. Since we are in the core execution regime (at most $f$ equivocators), no \voteqc forks occur; therefore progress continues indefinitely. If this regime is violated, replicas enter fork recovery, and only Resilient Path safety is claimed.
\end{proof}

\subsection{Resilient Path and fork recovery safety}

The Resilient Path provides an extended safety guarantee against an adversary capable of breaking the core $f$ Byzantine limit by leveraging up to $f_{abc}$ alive-but-corrupt replicas.

\begin{lemma}[Unique \checkpointqc per Height]
    \label{lem:unique-checkpointqc}
    If $f_{abc} < n - 3f - 2c$, no two conflicting \checkpointqcs can be formed for the same slot height $s$.
\end{lemma}
\begin{proof}
    By the protocol (Step~2, \cref{step:2}), an honest replica broadcasts at most one \textsc{ChkProp} per slot height $s$. Let $C_1$ and $C_2$ be two \checkpointqcs for different state roots at height $s$. Each requires $q = n-f-c$ distinct signatures. Because honest replicas do not double-sign \textsc{ChkProp}s at the same height, the intersection between the signers of $C_1$ and $C_2$ consists entirely of Byzantine or AbC replicas. The number of intersecting signers is at least $2q - n$. Thus, for two \checkpointqcs to form, the adversary must control at least $2q - n$ nodes. We require $f+f_{abc} < 2q - n$.

    We know $q = n-f-c$, so $2q - n = 2(n-f-c) - n = n - 2f - 2c$.
    Given our assumption that $f_{abc} < n - 3f - 2c$, we have:
    \[ f + f_{abc} < f + (n - 3f - 2c) = n - 2f - 2c = 2q - n \]
    This mathematically prohibits the adversary from successfully forming two \checkpointqcs, proving that at most one can exist per height.
\end{proof}

\cref{lem:unique-checkpointqc} rules out conflicts at a single height. To safely adopt the \emph{highest} checkpoint during recovery, we additionally need checkpoints at different heights to agree on their common prefix. Recall (\Cref{sec:protocol}) that state roots bind history: $\sigma_s$ commits to $\sigma_{s-1}$ and the applied digest. We call a \checkpointqc \emph{recorded} if at least one correct replica stores it from the moment it forms; in particular, any \checkpointqc for which some correct replica issued a \textsc{ChkWitness} is recorded, since Step~3 (\cref{step:3}) requires observing the certificate. A \checkpointqc assembled privately by the adversary from honest \textsc{ChkProp}s need not be recorded and may be orphaned by fork recovery; this affects no client, as Resilient-Path finalization requires a \finalityqc (\cref{thm:resilient-safety}).

\begin{lemma}[Checkpoint Chain Consistency]
    \label{lem:checkpoint-chain}
    If $f_{abc} < n - 3f - 2c$, then for any recorded \checkpointqc $C'$ at height $s'$ and any \checkpointqc $C$ at height $s > s'$, the history committed by $C$ contains the content $(h', \sigma')$ certified by $C'$ at height $s'$. Hence all recorded \checkpointqcs lie on a single chain.
\end{lemma}
\begin{proof}
    The signer sets of $C$ and $C'$ each have size $q = n-f-c$ and intersect in at least $2q - n = n - 2f - 2c$ replicas. Since $f + f_{abc} < n - 2f - 2c$, at least one common signer $R$ is non-equivocating ($R$ may be crash-faulty, but was live when signing both certificates). Non-equivocating replicas execute slots in order on an append-only local history, which changes only when fork recovery adopts the canonical branch (\cref{alg:fork_recovery}); within normal operation, \cref{lem:view-change-invariance} pins any slot with a \voteqc to its digest across views, and $R$ halts upon detecting a fork. Between recovery epochs, therefore, $R$'s history at height $s'$ is fixed at the content $(h', \sigma')$ it checkpoint-proposed when signing $C'$. Across recovery epochs the content is preserved because $C'$ is recorded: its correct holder inputs it to the Byzantine broadcast, so $C'$ is delivered in every recovery epoch after its formation. By strong induction on the number of recovery epochs preceding the formation of $C$, the highest delivered \checkpointqc $C_{max}$ of each such epoch satisfies $C_{max}.s \ge s'$ and its history contains $(h', \sigma')$ at height $s'$ (apply the lemma inductively to the pair $(C', C_{max})$, whose formation is preceded by strictly fewer recovery epochs); hence the adopted canonical branch---and every honest history thereafter---retains $(h', \sigma')$ at height $s'$. Since $R$ signed $C$, the root $\sigma$ certified by $C$ was computed on $R$'s history and therefore commits $(h', \sigma')$ at height $s'$. By collision resistance of the state commitment, every history consistent with $C$ agrees with $C'$ at height $s'$.
\end{proof}

\begin{theorem}[Resilient Path Safety]
    \label{thm:resilient-safety}
    If clients follow the Resilient Path (finalizing only upon seeing a \finalityqc), the states they finalize form a single prefix-consistent history---no two clients commit conflicting states at any height---even across view changes, synchronous recoveries, and post-recovery epochs, provided $f_{abc} < n - 3f - 2c$.
\end{theorem}
\begin{proof}
    A client finalizes a state only if it observes a \finalityqc, which requires an underlying valid \checkpointqc. By \cref{lem:unique-checkpointqc}, all \checkpointqcs at height $s$ certify the same content $(h,\sigma)$. Since honest replicas only broadcast a \textsc{ChkWitness} matching this uniquely determined checkpoint content (and only if its payload executes to the correct state), no \finalityqc can be formed for a conflicting state. Moreover, every checkpoint a client finalizes is recorded: a \finalityqc carries $q$ \textsc{ChkWitness} signatures, of which at least $q - f - f_{abc} - c \ge f+1$ belong to correct replicas, and a correct replica issues a \textsc{ChkWitness} only after observing---and storing---the underlying \checkpointqc (Step~3, \cref{step:3}).

    During Fork Recovery (\cref{alg:fork_recovery}), replicas exchange all known \checkpointqcs via a Byzantine Broadcast protocol (e.g., Dolev-Strong). Because Dolev-Strong does not rely on an honest majority, it guarantees identical outputs for all correct nodes despite a potentially dishonest majority ($f+f_{abc} > n/2$). Alternatively, if $f+f_{abc} < n/2$, an honest-majority broadcast can be deployed to terminate much faster (e.g., in expected constant rounds~\cite{Katz2006Expected,Abraham2020Sync}). In either case, all correct replicas will obtain the exact same set $\mathcal{C}$ of \checkpointqcs, and $\mathcal{C}$ contains every recorded \checkpointqc---in particular every client-finalized one. By \cref{lem:unique-checkpointqc} at most one \checkpointqc exists per height, and by \cref{lem:checkpoint-chain} the \checkpointqc with the highest slot $C_{max}$ extends the content certified by every recorded \checkpointqc at lower heights. All correct replicas deterministically adopt the branch of $C_{max}$, which therefore preserves any state finalized by a Resilient Path client prior to the fork.
\end{proof}
\section{\orcdag: A DAG-Based Implementation}
\label{sec:dag-application}

\Cref{sec:model,sec:proofs} are independent of how
proposals are represented: safety uses first-round votes and the
thresholds $q$ and $k$ of \cref{thm:main}. The
\sysname protocol (\Cref{sec:protocol}) uses $q$ for every
certificate round. Many
high-throughput systems nonetheless realize the same logic over a
\emph{directed acyclic graph} of blocks or vertices, with edges for
dependencies and leaders extending a frontier of the DAG~\cite{narwhal,bullshark,mysticeti,mahi-mahi,sailfish,shoal,hammerhead}. We present \orcdag, a DAG-based instantiation of \sysname.

\begin{algorithm*}[t]
    \caption{\orcdag Instantiation (for leader block $B$ at round $R$)}
    \label{alg:dag}
    \algsize
    \begin{algorithmic}[1]
        \State \textbf{Definitions:}
        \Statex \indent \textbf{Support}: A block $b'$ \emph{supports} $B \equiv (A,R,h)$ if $B$ is the first block by author $A$ in round $R$ encountered in the deterministic depth-first traversal of $b'$'s causal history~\cite{mysticeti}.
        \Statex \indent \textbf{Vote}: A block in round $R+1$ that supports $B$.
        \Statex \indent \textbf{Blame}: A block in round $R+1$ that does \emph{not} support $B$.
        \Statex
        \State \textbf{Direct Decision Rule} (evaluated when round $R+1$ blocks are delivered):
        \State \textbf{if} blocks from $\geq q$ distinct authors in round $R+1$ vote for $B$ \textbf{then} \textsf{Commit}($B$)
        \State \textbf{if} blocks from $\geq q$ distinct authors in round $R+1$ blame $B$ \textbf{then} \textsf{Skip}($B$)
        \Statex
        \State \textbf{Indirect Decision Rule} (evaluated when anchor $A$ is committed at round $\geq R+2$):
        \State \textbf{if} $B$ is undecided \textbf{then}
        \State \indent \textbf{if} $A$ links to blocks from $\geq k$ distinct authors in round $R+1$ that vote for $B$ \textbf{then}
        \State \indent\indent \textsf{Commit}($B$)
        \State \indent \textbf{else}
        \State \indent\indent \textsf{Skip}($B$)
    \end{algorithmic}
\end{algorithm*}

A DAG vertex (block) plays the role of the primary's proposal for a slot; a replica issues \textsf{accept}/\textsf{reject} once its local validity rules (including dependency checks) hold. The key idea of \orcdag is that a vote is not a separate signed message but a causal edge: a round $R+1$ block votes for a block $B$ precisely when it \emph{supports} $B$ in the sense of Mysticeti~\cite{mysticeti}---$B$ is the first block by its author for round $R$ encountered in a deterministic depth-first traversal of the voter's causal history---and blames $B$ when it does not (the Vote and Blame of \cref{alg:dag}). Support matters because causal histories merge: a voter's history may contain several equivocating copies of a leader block, but support selects at most one block per author-round slot. Hence a non-equivocating author casts exactly one vote per slot, and every \orcdag execution maps to an execution of \sysname, so the safety results of \Cref{sec:proofs} apply verbatim. The DAG structure itself carries the first-round votes, so no explicit vote messages are needed even though the DAG is uncertified. The Direct Decision Rule of \cref{alg:dag} realizes the Fast-Path: collecting $q$ agreeing first-round votes is exactly $\geq q$ distinct round $R+1$ authors voting for $B$, using the same quorum $q = n-f-c$ from \cref{thm:main}; symmetrically, $q$ blames force a \textsf{Skip}.

The Indirect Decision Rule of \cref{alg:dag} realizes the View-Change. The ``later authenticated artifact'' of the abstract protocol \sysname is concretely a committed anchor $A$ at a round $\geq R+2$: because reachability in the DAG certifies the votes enclosed in its causal past, $A$ linking to $\geq k$ distinct round $R+1$ authors that vote for $B$ plays the role of the embedded vote reports in the view change of \Cref{sec:protocol}, with the accept threshold $k = 2f+c+1$. This single anchor object subsumes the abstract alternatives---such as a commit certificate on a descendant, a checkpoint quorum, or a bundle hashing prior votes---as all collapse to the anchor's causal history, which reachability certifies without any separate certificate. The choice of $k$ is exactly what the Reachability, Exclusivity, and Liveness conditions of \cref{sec:proofs} require for safe indirect commits.

\Cref{sec:algorithms} provides detailed algorithms to formally define \orcdag.


\section{Evaluation} \label{sec:evaluation}

\orcdag is a DAG-based instantiation of \sysname{} that we implement and benchmark, built in Rust as a fork of Mysticeti~\cite{mysticeti-code}. \Cref{sec:implementation} details its implementation and testing methodology.

\subsection{Evaluation scope}

Our evaluation has a single goal: quantify the latency we gain by trading away fault tolerance, i.e., by moving from the optimal $n=3f+1$ three-message-delay fault budget to \sysname's $n=5f+3c+1$ two-message-delay design space. Put differently: \emph{what latency do we gain by giving up the fault tolerance represented by the orange wedge of \Cref{fig:fault-coverage} (\Cref{sec:intro})?} This is not trivial. A naive observer might expect a direct 33\% latency reduction simply from cutting one of the three commit rounds. But two factors complicate the picture: (1) \sysname needs a larger quorum of replicas per round than a typical $n=3f+1$ protocol; and (2) the saved round only reduces the latency from block proposal to commit), not \emph{queuing latency} (from transaction submission to block proposal).

Our evaluation makes the following claims:
\begin{itemize}
    \claimitem{claim:c1} \orcdag's lower latency is bought with fault tolerance alone, not throughput.
    \claimitem{claim:c2} Under benign crash faults, \orcdag degrades gracefully, retaining its throughput and low latency.
    \claimitem{claim:c3} \orcdag's ordering-latency reduction over the $n=3f+1$ baseline is ${\sim}24\%$ at low load and erodes only mildly (to ${\sim}22\%$) at high load.
    \claimitem{claim:c4} \orcdag's ordering latency is independent of how the fault budget splits between $f$ and $c$ in a typical geo-distribution, and a split budget deploys at strictly smaller $n$ than the pure-Byzantine $5f+1$ ($c=0$) point.
    \claimitem{claim:c5} Geo-distribution is decisive: \orcdag's latency advantage holds when all protocols operate over the same geo-locations, but is negated when the baseline's smaller quorum can exclude a remote region that \orcdag's larger quorum cannot.
    \claimitem{claim:c6} The checkpoint engine is virtually overhead-free: enabling it changes neither \orcdag's throughput nor its ordering latency.
    \claimitem{claim:c7} End-to-end latency is ${\sim}2\times$ the ordering latency for \orcdag and the baseline alike, and \orcdag's end-to-end latency is ${\sim}20$--$25\%$ below the baseline's across all loads and $(f,c)$ splits.
\end{itemize}

Benchmarking BFT protocols under actual Byzantine behavior is an open problem~\cite{twins,dos-resilience}; the state of the art establishes worst-case guarantees through formal proofs, which we give in \Cref{sec:proofs}.

\subsection{Experimental setup}

To demonstrate these claims, we deploy four \orcdag configurations spanning the fault-budget spectrum (\Cref{fig:param-coverage}): \emph{Byzantine-only} ($f{=}10,c{=}0$), \emph{Byzantine-heavy} ($f{=}8,c{=}3$), \emph{Balanced} ($f{=}6,c{=}6$), and \emph{Crash-heavy} ($f{=}2,c{=}13$).
We use two baselines:
\emph{(i)} Mysticeti~\cite{mysticeti}, the closest $n=3f+1$ protocol to \orcdag in terms of both design and implementation, representing the standard 3-round PBFT-style baseline;
\emph{(ii)} Hydrangea~\cite{Hydrangea}, an alternative hybrid $(f,c)$ design with $n=3f+2c+k+1$ that can commit in a single round-trip (via an optimistic path) and is, to our knowledge, the only other protocol in this design space with a deployable implementation. The pure-Byzantine $c=0$ point (an $n=5f+1$ deployment) is itself an \orcdag configuration, which we include as an additional reference rather than as a separate system.

\begin{figure}[t]
    \centering
    \includegraphics[width=\linewidth]{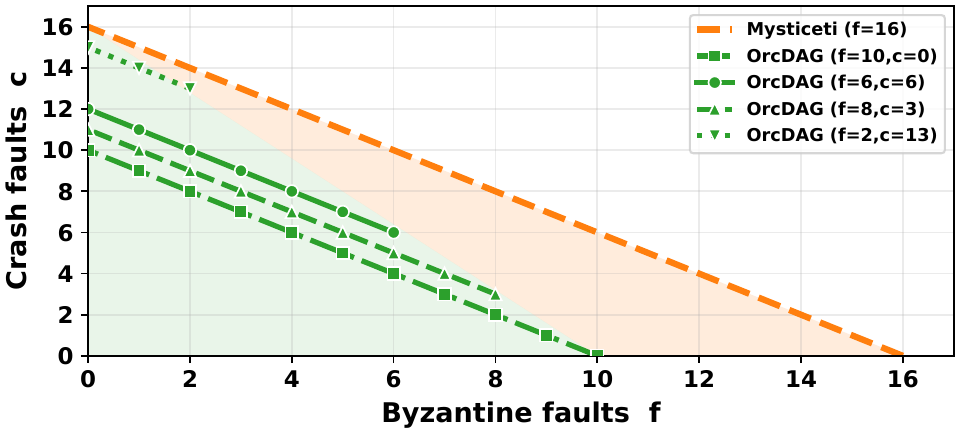}
    \caption{\footnotesize Runtime fault tolerance of the benchmarked configurations.}
    \label{fig:param-coverage}
\end{figure}

Several recent protocols explore designs in the same neighborhood (\Cref{sec:related}), including Minimmit~\cite{Minimmit}, Kudzu~\cite{Kudzu}, and Alpenglow~\cite{Alpenglow}. None of these provide deployable implementations with networking code; their codebases are intended for simulation only and are therefore excluded from our WAN measurements. The remaining systems, Mysticeti and Hydrangea, both have mature deployable implementations and serve as our baselines. We configure Hydrangea at $f{=}9$, $c{=}10$, $k{=}2$ for the $n{=}50$ committee; in this regime it tolerates more crash faults than \orcdag, but, being built atop HotStuff, it inherits the data-dissemination bottleneck identified by Narwhal~\cite{narwhal}.

Unless stated otherwise, we emulate a typical blockchain replica distribution~\cite{solana-decentralization,suiscan}: a fast quorum spanning several EU and US regions, plus a remote tail (Tokyo).
We report median (p50) latencies with p90 whiskers. \Cref{sec:experimental-setup} describes the precise geo-distribution and testbed details used in this section.
In all graphs, \emph{ordering latency} refers to the time elapsed from the moment a client submits a transaction to when it is ordered by consensus (including queuing); \emph{checkpoint latency} refers to the additional time until the transaction is covered by a \emph{certified checkpoint}, that is, a commitment to the executed state attested by a quorum (\cref{sec:implementation}); \emph{end-to-end latency} is their sum; and \emph{throughput} refers to the number of (512 bytes) transactions processed per second.

\subsection{Throughput vs.\ fault tolerance trade-off}

\Cref{fig:happy-case} evaluates a roughly $50$-replica WAN deployment under failure-free conditions. The throughput of \orcdag matches Mysticeti, which is expected as both build upon the same uncertified-DAG fabric, corroborating related work~\cite{narwhal}. For cost reasons, we cap the input load at $100{,}000$\,tx/s, which is two orders of magnitude above the peak throughput observed in deployed blockchains and ample to stress the systems~\cite{sui-lutris}. This confirms claim~\ref{claim:c1}: \orcdag is realizable without sacrificing throughput.
Hydrangea performs significantly worse at scale (red line in \cref{fig:happy-case}). At $50$ replicas, its latency exceeds the plotted range even under minimal load, demonstrating that it does not scale gracefully to large committees. This limitation stems from its leader-driven data-dissemination mechanism, consistent with findings reported in related work~\cite{narwhal}.

The figure also shows that \orcdag orders at lower latency than Mysticeti; we dissect this latency--fault-tolerance trade-off in \Cref{sec:latency-tradeoff}.

\begin{figure}[t]
    \centering
    \includegraphics[width=\linewidth]{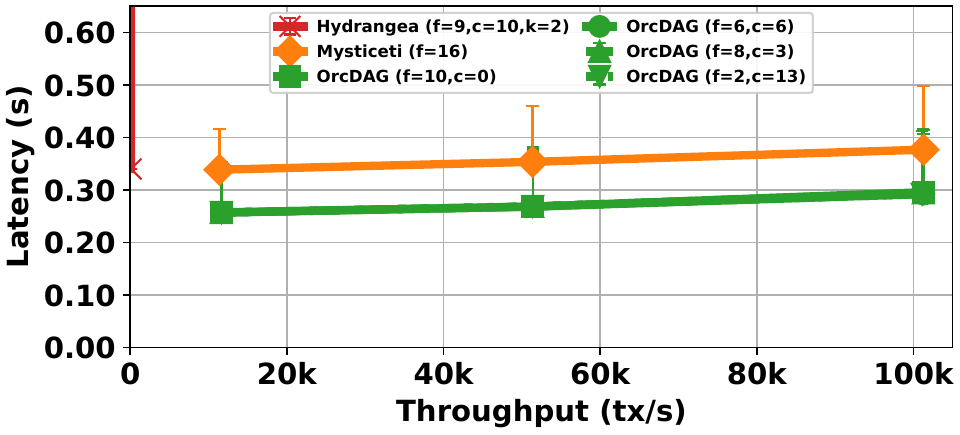}
    \caption{\footnotesize Throughput vs ordering latency, no faults, $50$ replicas.}
    \label{fig:happy-case}
\end{figure}

\subsection{Impact of benign crash faults}

\Cref{fig:faults} shows \orcdag, Mysticeti, and Hydrangea under benign crash faults in committees of approximately $10$ replicas. The DAG-based systems sustain the offered load with a graceful latency inflation relative to the fault-free runs, absorbing crashes by quickly skipping crashed leaders via the direct skip rule (\Cref{sec:dag-application}). Each system runs at the minimal committee for this fault budget---Mysticeti at $n{=}7$ ($f{=}2$), \orcdag at $n{=}9$ ($f{=}1,c{=}1$), and the $c{=}0$ reference at $n{=}11$ ($f{=}2$)---so every committee operates at its maximum fault load. In these benign fault conditions, \orcdag matches its own $c=0$ point because crashed replicas count against $c$ rather than $f$ and so do not consume the Byzantine budget (the $n{=}9$ hybrid and the $n{=}11$ $c{=}0$ reference commit within $1$\,ms of each other: $378$ versus $379$\,ms at low load, $395$ versus $394$\,ms at high load). This split remains latency-neutral while delivering a latency win over Mysticeti: at low load ($10{,}000$\,tx/s) \orcdag commits in $378$\,ms versus $492$\,ms (a ${\sim}23\%$ reduction), and near the high-load cap ($50{,}000$\,tx/s) $395$\,ms versus $567$\,ms (a ${\sim}30\%$ reduction). This confirms claim~\ref{claim:c2}.
Hydrangea's latency exceeds the plotted range, confirming it does not tolerate faults gracefully, corroborating related work~\cite{narwhal}. These small committees represent the worst case for the fault budget: a crash removes a large fraction of a $10$-replica committee but a negligible fraction of a $50$-replica one.

\begin{figure}[t]
    \centering
    \includegraphics[width=\linewidth]{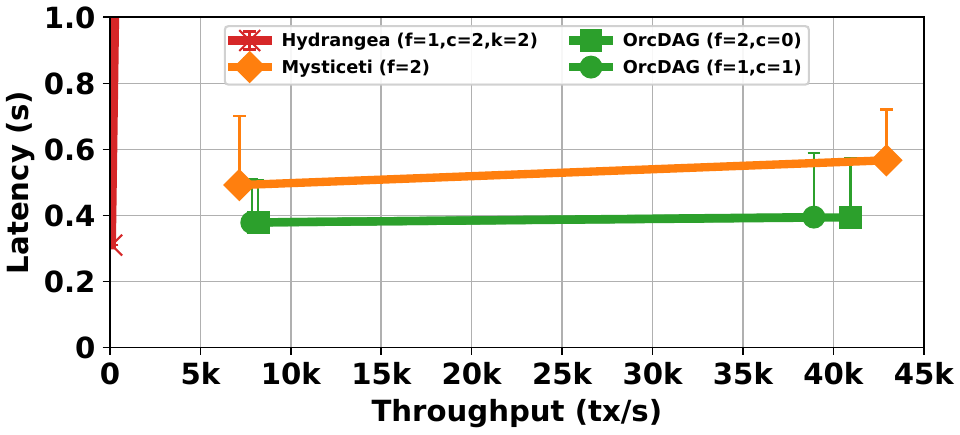}
    \caption{\footnotesize Throughput vs ordering latency, $2$ crash faults, minimal committees.}
    \label{fig:faults}
\end{figure}

\subsection{Understanding the latency trade-off}\label{sec:latency-tradeoff}

We zoom into \Cref{fig:happy-case} to understand \orcdag's ordering-latency--fault-tolerance trade-off. At $100$k\,tx/s, \orcdag orders in $292$--$296$\,ms (p50) versus $377$\,ms for Mysticeti, a ${\sim}22\%$ reduction; at low load ($10$k\,tx/s), the gap is $257$\,ms versus $339$\,ms, a ${\sim}24\%$ reduction. All \orcdag configurations, including the $c{=}0$ ($5f{+}1$) reference, agree within a few milliseconds at both low and high load. A naive view expects the move from three to two message delays to yield a ${\sim}33\%$ reduction in ordering latency. The measured ${\sim}24\%$ reduction is explained by the shared, load-independent base: queuing latency's constant component---since replicas propose a block every round, a transaction waits on average half a round for inclusion---is identical for both systems, meaning the saved round represents a smaller fraction of the total. Moreover, the quorum-width difference (${\sim}80\%$ of replicas for \orcdag versus ${\sim}67\%$ for Mysticeti) costs little in this geo-distribution because both quorums form within the same set of regions (see \Cref{sec:quorum-location}).

\Cref{fig:bars-10k} reports per-protocol ordering latency at low load ($10$k\,tx/s), which isolates the quorum-width effect because load-induced queuing is negligible. Here, the ${\sim}24\%$ reduction holds uniformly across all \orcdag configurations.
\Cref{fig:bars-100k} adds high load ($100$k\,tx/s), introducing load-induced queuing. This steady-state effect (latency stays flat rather than accumulating) arises because larger blocks prolong transmission and higher per-node jitter inflates the wait for the last quorum member. This symmetric queuing adds $+38$\,ms for Mysticeti and $+35$--$39$\,ms for the \orcdag configurations. Because additive queuing is common while base ordering latencies differ, the relative reduction erodes mildly, from ${\sim}24\%$ to ${\sim}22\%$ as load grows from $10$k to $100$k\,tx/s, similarly across all $(f,c)$ splits. \Cref{fig:improvement} plots the total ordering-latency reduction over Mysticeti against load; its near load-invariance confirms claim~\ref{claim:c3}.

\begin{figure}[t]
    \centering
    \includegraphics[width=\linewidth]{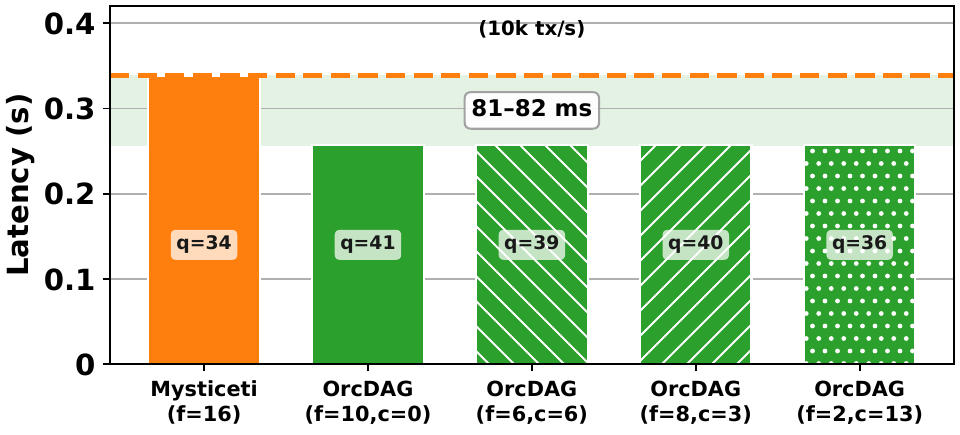}
    \caption{\footnotesize Per-protocol ordering latency at low load ($10$k\,tx/s); load-induced queuing negligible.}
    \label{fig:bars-10k}
\end{figure}

\begin{figure}[t]
    \centering
    \includegraphics[width=\linewidth]{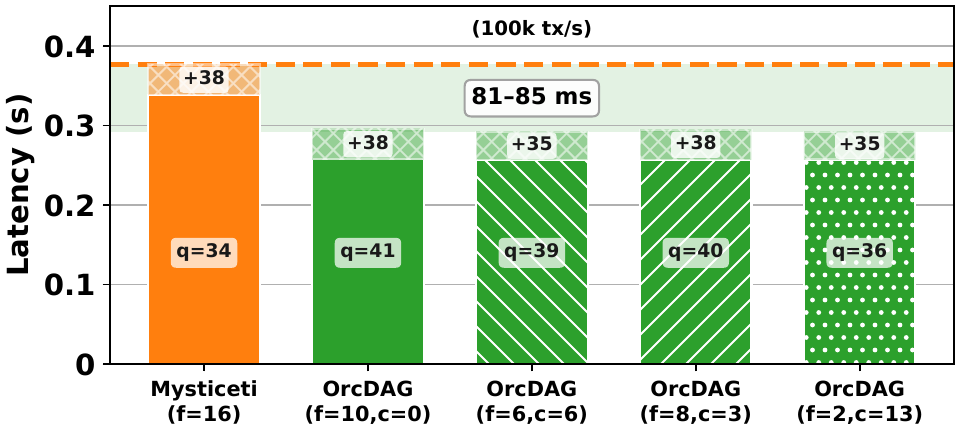}
    \caption{\footnotesize Per-protocol ordering latency at high load ($100$k\,tx/s); load-induced queuing segment stacked.}
    \label{fig:bars-100k}
\end{figure}

\begin{figure}[t]
    \centering
    \includegraphics[width=\linewidth]{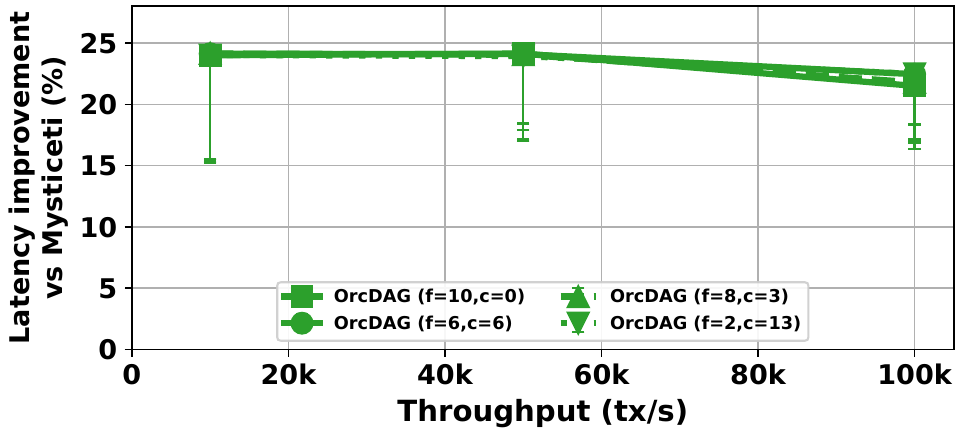}
    \caption{\footnotesize Ordering latency improvement over Mysticeti vs load.}
    \label{fig:improvement}
\end{figure}

\subsection{The role of quorum location} \label{sec:quorum-location}

\Cref{fig:bars-10k,fig:bars-100k} show that, across \orcdag configurations, ordering latency is not ordered by quorum size. All configurations order in two message delays over the same uncertified-DAG structure, so in a fault-free run their ordering latencies are effectively indistinguishable regardless of how the fault budget splits between $f$ and $c$. This confirms claim~\ref{claim:c4}: in the common case \orcdag pays no performance penalty for the hybrid-fault analysis, and a split budget deploys at strictly smaller $n$ than the pure-Byzantine $5f+1$ ($c=0$) point. Concretely, the configurations cluster at $257$\,ms (within $1$\,ms) at low load and $292$--$296$\,ms at high load, versus $339$ and $377$\,ms for Mysticeti; the $c=0$ (Byzantine-only) point is among the fastest, so the $0$--$4$\,ms spread is run-to-run measurement noise (WAN and egress jitter). Crucially, this neutrality is conditional. The split sets the quorum $q = n - f - c$; the Byzantine-leaning configuration carries a larger quorum ($q{=}41$ for $c{=}0$) than the crash-leaning one ($q{=}36$ for Crash-heavy). The configurations are latency-neutral only because, in this graceful geo-distribution, all of these quorums still form in the EU-US without reaching the remote tail in Tokyo.

\begin{figure}[t]
    \centering
    \includegraphics[width=\linewidth]{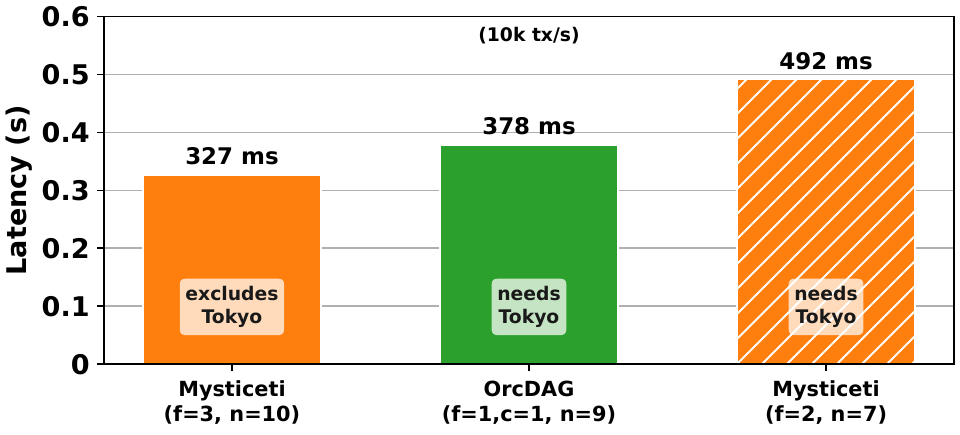}
    \caption{\footnotesize Latency when the quorum is forced to include the remote tail.}
    \label{fig:geo-boundary}
\end{figure}

\Cref{fig:geo-boundary} shows what happens when the quorum is forced to reach beyond the EU-US. We use crash faults on minimal committees ($2$ crashes) to force the quorum to include the remote tail (Tokyo). When every protocol must reach Tokyo, Mysticeti is penalized more because its extra (third) round waits on the remote replica once more than \orcdag. But when Mysticeti's larger fault-tolerance budget gives it enough slack to exclude Tokyo while \orcdag's tighter quorum cannot, the advantage reverses. Concretely, Mysticeti with $f{=}3$ ($n=10$) keeps slack and excludes Tokyo ($327$\,ms), whereas \orcdag with $f{=}1,c{=}1$ ($n=9$) and Mysticeti with $f{=}2$ ($n=7$) have no slack and must include Tokyo ($378$\,ms and $492$\,ms). So \orcdag beats the minimal Mysticeti (two rounds versus three) but loses to the over-provisioned Mysticeti that dodges Tokyo. This confirms claim~\ref{claim:c5}. The same mechanism bounds claim~\ref{claim:c4}: under a less graceful geo-distribution, a Byzantine-leaning configuration's larger quorum ($q{=}41$) could be forced to include the remote tail while a crash-leaning one ($q{=}36$) avoids it, at which point the split would no longer be latency-neutral.

\subsection{The cost of checkpoints}\label{sec:eval-checkpoints}

\begin{figure}[t]
    \centering
    \includegraphics[width=\linewidth]{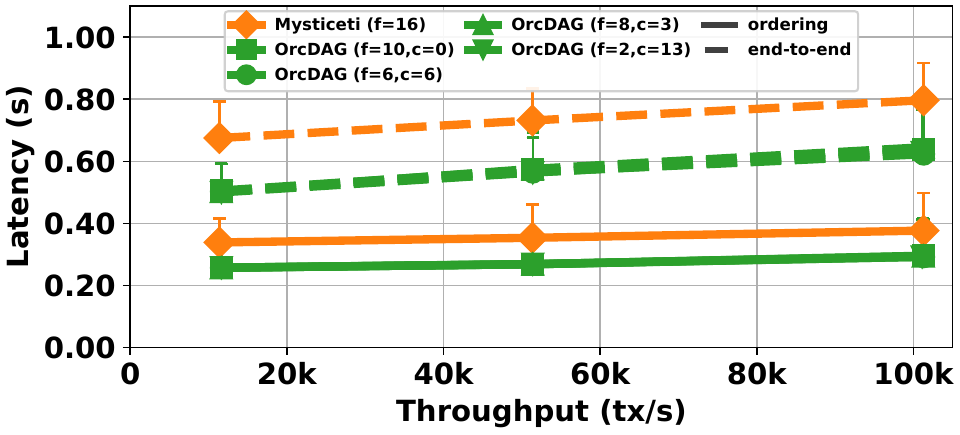}
    \caption{\footnotesize Ordering (solid) and end-to-end (dashed) latency, no faults, ${\sim}50$ replicas.}
    \label{fig:ckpt-happy-case}
\end{figure}

\begin{figure}[t]
    \centering
    \includegraphics[width=\linewidth]{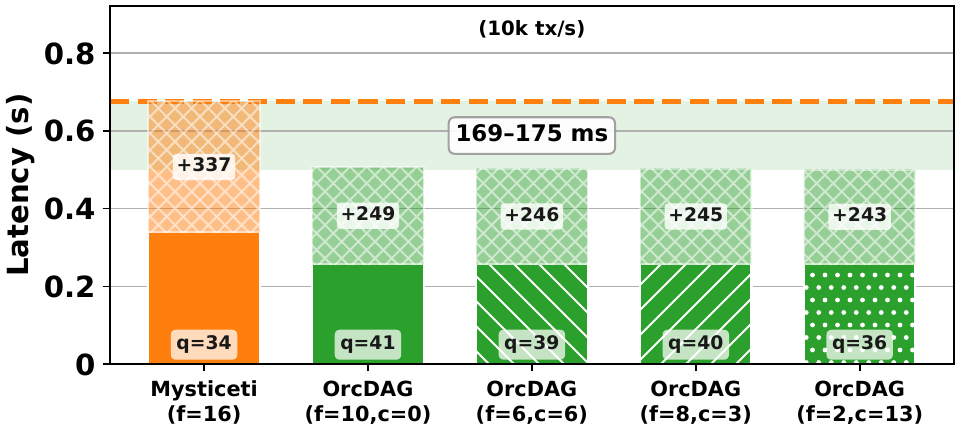}
    \caption{\footnotesize Latency breakdown at low load ($10$k\,tx/s): ordering (solid) plus checkpoint (hatched) latency.}
    \label{fig:ckpt-bars-10k}
\end{figure}

\begin{figure}[t]
    \centering
    \includegraphics[width=\linewidth]{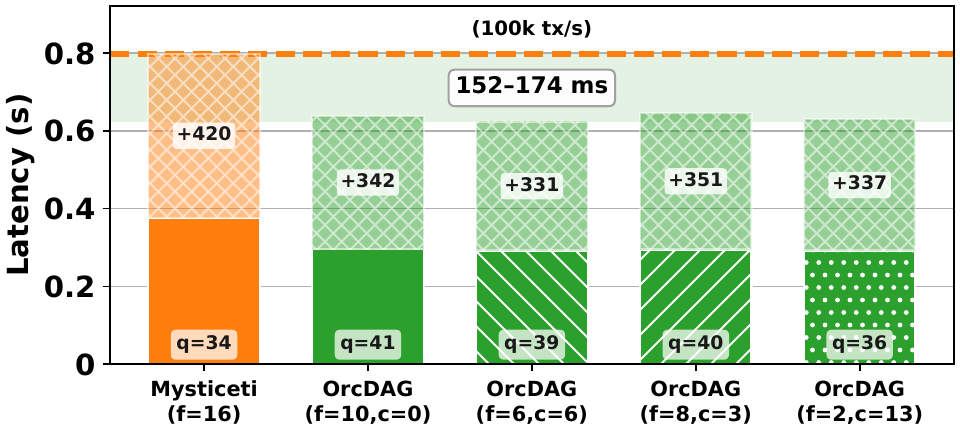}
    \caption{\footnotesize Latency breakdown at low load ($100$k\,tx/s): ordering (solid) plus checkpoint (hatched) latency.}
    \label{fig:ckpt-bars-100k}
\end{figure}

\begin{table}[t]
    \caption{Ordering latency (p50\,/\,p90, ms) with the checkpoint engine enabled and disabled (\orcdag, $f{=}6,c{=}6$).}
    \label{tab:ckpt-onoff}
    \centering\small
    \begin{tabular}{@{}lcc@{}}
        \toprule
        \textbf{Load} & \textbf{Engine enabled} & \textbf{Engine disabled} \\
        \midrule
        $10$k\,tx/s   & $257$ / $353$           & $253$ / $350$            \\
        $100$k\,tx/s  & $292$ / $407$           & $290$ / $404$            \\
        \bottomrule
    \end{tabular}
\end{table}

We implement the checkpoint engine (\cref{sec:implementation}) emulating the design of production blockchains~\cite{sui-code,aptos-code}. It finalizes checkpoints by ordering replica attestations through consensus itself. That is, each replica builds an \emph{attestation}---a small transaction containing a commitment to its executed state---and submits it to the consensus protocol like any other transaction. A checkpoint is \emph{certified} once a quorum of matching attestations is ordered.
For \orcdag, a certified checkpoint carries the Resilient-Path guarantee (\finalityqc, \cref{sec:protocol}); for the baseline, it provides durability only. We first establish that running the engine affects neither ordering latency nor throughput. As shown in \cref{tab:ckpt-onoff}, with the checkpoint engine disabled, \orcdag ($f{=}6,c{=}6$) orders in $253$\,ms p50 ($350$\,ms p90) at $10$k\,tx/s and $290$\,ms p50 ($404$\,ms p90) at $100$k\,tx/s. Enabling the engine yields ordering latencies of $257$\,ms p50 ($353$\,ms p90) and $292$\,ms p50 ($407$\,ms p90), respectively. The minor $2$--$4$\,ms differences are well within standard run-to-run WAN and egress jitter (${\sim}10$--$15$\,ms). Throughput remains unchanged as both configurations absorb the full offered load. The stream of attestation transactions adds only ${\sim}1{,}300$ observations/s ($25$ checkpoints per second $\times$ $51$ replicas), representing a negligible ${\sim}1.3\%$ traffic overhead over the main transaction stream at $100$k\,tx/s (and ${\sim}14\%$ at $10$k\,tx/s, still with no measurable impact on latency). This confirms claim~\ref{claim:c6}.

As shown in \Cref{fig:ckpt-happy-case}, end-to-end latency runs parallel to ordering latency for both systems, costing roughly one additional ordering cycle (end-to-end latency is ${\sim}2.0$--$2.2\times$ the ordering latency). This is because certification requires a quorum of attestations to be ordered, effectively constituting a second ordering cycle. Two structural factors cause checkpoint latency to slightly exceed a single median ordering cycle: first, a checkpoint is certified only when the $q$-th attestation is ordered (which tracks a high order statistic of the attestation-ordering distribution rather than the median); second, transactions must wait for the next periodic checkpoint boundary (up to ${\sim}40$\,ms). State execution and attestation formation contribute a further negligible $1$--$8$\,ms p50. \Cref{fig:ckpt-bars-10k,fig:ckpt-bars-100k} show that checkpoint latency is $243$--$249$\,ms at $10$k\,tx/s and $331$--$351$\,ms at $100$k\,tx/s for \orcdag, compared to $337$\,ms and $420$\,ms for Mysticeti. Because \orcdag saves one message delay in both the ordering of transactions and the ordering of attestations, its latency advantage compounds: its end-to-end latency is $169$--$175$\,ms (${\sim}25\%$) below Mysticeti's at low load, and $152$--$174$\,ms (${\sim}20\%$) below at high load. The latency-neutrality of the $(f,c)$ split (claim~\ref{claim:c4}) also extends to end-to-end latency, with all four configurations agreeing within a few milliseconds at all load points. Strikingly, \orcdag's end-to-end latency at $100$k\,tx/s ($623$--$645$\,ms) is lower than Mysticeti's at $10$k\,tx/s ($675$\,ms). This confirms claim~\ref{claim:c7}. We note that because our implementation certifies checkpoints by ordering attestations through consensus, the four-delay bound of \cref{sec:protocol} is a lower bound that a DAG-native implementation could approach by certifying attestations once they are causally covered by a quorum of later blocks, rather than waiting for their commit.

\section{Related Work}
\label{sec:related}

\paragraph*{Hybrid fault models.}
The formal separation of Byzantine and benign faults was first explored by Thambidurai and Park~\cite{Thambidurai1988} for interactive consistency. While their work demonstrated that protocols can achieve higher resilience by not treating all failures as worst-case Byzantine, it did not address the latency limitations of state machine replication (SMR). Later pragmatic systems like UpRight~\cite{Clement2009UpRight} and XFT~\cite{Liu2016XFT} applied this separation to cluster architectures. However, UpRight focuses on end-to-end service robustness rather than theoretical latency minimums, and XFT opportunistically tolerates Byzantine faults only when an honest majority communicates synchronously. In contrast, our work fundamentally re-examines the quorum intersections of SMR under partial synchrony, specifically optimizing for the minimal 2-message-delay commit path while treating $f$ and $c$ as distinct variables.

\paragraph*{Low-latency SMR and $5f{+}1$ protocols.}
The demand for ultra-low latency in decentralized networks has driven recent protocols to instantiate optimal two-round decision paths. Minimmit~\cite{Minimmit} relies on a pure counting model that conservatively requires $n \geq 5f+1$, treating any offline replica as a potential equivocator. Alpenglow~\cite{Alpenglow} refines this regime with a separate allowance for offline stake (its ``20+20'' model), but its safety holds only when leaders do not equivocate, so it does not withstand the unrestricted Byzantine behavior we consider. Kudzu~\cite{Kudzu} integrates an optimistic fast path: with $n = 3f+2p+1$ replicas, it commits in two rounds only when at most $p$ replicas actually misbehave, reverting to a slower three-round path otherwise. Hydrangea~\cite{Hydrangea} brings this optimistic-fast-path design to a hybrid $(f,c)$ analysis requiring $n = 3f+2c+k+1$: its two-round commit is optimistic in the same sense---the fast path remains live only while the number of actual faults is at most a threshold $p$, beyond which the protocol falls back to a slower, PBFT-style 3-round path. Pushed to its extreme configuration ($k=2f+c-3$, hence $n=5f+3c-2$), this optimistic threshold shrinks to $p=f+c-2$, strictly below the full fault budget $f+c$ the system is configured to tolerate. Hydrangea's smaller committee therefore does not contradict our lower bound: it dips below $5f+3c+1$ only by conditioning fast-path liveness on fewer actual faults than the configured budget.
Furthermore, when Hydrangea fails to commit on a leader, it requires at least five message delays to recover and commit that leader's transactions, whereas \sysname can complete this in only four. Additionally, unlike \sysname, Hydrangea lacks a mechanism to rapidly skip unresponsive leaders in the style of Mysticeti~\cite{mysticeti}, making its recovery potentially slower still.
The two-step lineage dates to FaB Paxos~\cite{FaBPaxos}, which established $n \ge 5f+1$ with a matching lower bound for certificate-forwarding recovery; a theoretical construction can shave this to $5f-1$~\cite{Kuznetsov2021Revisiting}, at a complexity cost that has kept practical systems at $5f+1$. Unlike these approaches, we prove that $n = 5f+3c+1$ is necessary for any protocol that commits in two message delays from a single vote quorum and remains live under the full budget of $f$ Byzantine and $c$ crash faults (\cref{thm:main}), and sufficient via \sysname---the three extra replicas over Hydrangea's extreme point buy an unconditional 2-message-delay commit that needs no fallback path, while separating the cost of crash faults from Byzantine faults. Our choice allows us to also provide a resilient path for clients to tolerate an additional $f_{abc}$ equivocators at the cost of two extra message delays, without degrading the core protocol's optimal liveness in executions where AbC replicas do not act, but results in lower fast-path liveness guarantees than Hydrangea.

\paragraph*{Client-side safety and Alive-but-Corrupt faults.}
FlexibleBFT~\cite{Malkhi2019FlexibleBFT} introduced the concept of \emph{alive-but-corrupt} ($f_{abc}$) faults, cleanly separating replica-quorum liveness assumptions from client-visible safety guarantees. While FlexibleBFT primarily uses this slack to separate synchronous and asynchronous network assumptions, we apply it directly to the 2-delay quorum framework. Unlike existing systems that force all clients to accept the same latency-security trade-off, our dual-path architecture explicitly empowers clients. By chaining \checkpointqcs, our Resilient Path allows clients to tolerate an additional $f_{abc}$ equivocators at the cost of two extra message delays; the core protocol's optimal liveness is unaffected as long as AbC replicas do not act, and if they do, liveness is only temporarily lost until synchronous recovery completes while resilient-path safety is preserved throughout. Orthogonally, a line of low-latency payment systems---FastPay~\cite{fastpay}, Zef~\cite{zef}, and Stingray~\cite{stingray}---forgoes consensus for single-owner transactions, trading general programmability for latency; \sysname instead retains full SMR..

\section{Conclusion}
\label{sec:conclusion}

We explored the fundamental limits of achieving optimal 2-message-delay consensus under a hybrid fault model. By cleanly separating Byzantine faults ($f$) from crash faults ($c$), we derived tight quorum intersections requiring $n \geq 5f+3c+1$ with an optimal Fast-Path threshold of $q = n-f-c$ ($4f+2c+1$ at the minimal committee size). This allows modern decentralized systems to relax the severe liveness requirements associated with pure $5f+1$ protocols without sacrificing latency.
Building upon these thresholds, we introduced \sysname, a dual-path consensus protocol that explicitly exposes a latency-safety trade-off to clients. Clients requiring ultra-low latency can finalize in two message delays via the Fast-Path (relying on \voteqcs), while clients prioritizing safety can wait four message delays for the Resilient Path (chaining \checkpointqcs and \finalityqcs). By enforcing a strict single-proposal rule for checkpoints, the Resilient Path guarantees safety against an extended set of $f_{abc}$ alive-but-corrupt replicas. If the core $f$ threshold is ever breached, the system safely halts and employs an Authenticated Byzantine Broadcast recovery mechanism to deterministically salvage the Resilient Path.
\ifpublish
    Finally, we demonstrated the practical applicability of our results by showing how they seamlessly map onto state-of-the-art DAG-based architectures. Future work includes extending these hybrid threshold derivations to fully asynchronous consensus environments and developing formal economic models to dynamically adjust $f$, $c$, and $f_{abc}$ budgets during live deployments.
\fi
\ifpublish
\section*{Acknowledgements}

This work is funded by Mysten Labs. We would like to thank Aniket Kate for discussions on early drafts.

\fi

\bibliographystyle{IEEEtran}
\bibliography{ref}

\appendices
\crefalias{section}{appendix}

\section{Detailed Algorithms for the DAG-Based Variant} \label{sec:algorithms}
This appendix complements \Cref{sec:dag-application} by formally defining the commit logic of \orcdag, the DAG-based variant of \sysname, through detailed algorithms.

\paragraph*{DAG-building layer.}
We assume the underlying DAG-building logic of Mysticeti~\cite{mysticeti}: replicas proceed in logical rounds; in each round every honest replica proposes one block referencing $\geq q$ distinct valid blocks from the previous round; blocks are disseminated to others; only blocks whose entire causal history has been validated are stored locally. The decision logic specified here operates on this local DAG and is independent of how blocks reach the replica.

\paragraph*{Entry point and idempotency.}
\Cref{alg:main} is the commit-logic entry point. Inline with related work~\cite{mysticeti,mahi-mahi}, it is idempotent and stateless from the DAG engine's perspective: the engine may invoke it whenever it likes, typically upon receiving and integrating a new block, passing the highest round currently in the local DAG ($r_{highest}$) and the round of the last block already committed ($r_{committed}$). The procedure returns the extension to the commit sequence (possibly empty) that the engine should append to its committed prefix.
The entry point is \Call{ExtendCommitSeq}{$r_{committed}, r_{highest}$}, which internally calls \Call{TryDecide}{} to evaluate each undecided leader slot using the rules in \Cref{alg:decider}, then linearises the causal sub-DAG of every newly committed leader (as introduced by DAG-Rider~\cite{dag-rider}). \Cref{alg:decider} specifies the per-slot decision process and with the supporting helper procedures (\Call{GetDecisionBlocks}{}, \Call{GetLeaderBlocks}{}, \Call{Link}{}).

\begin{algorithm}[t]
    \caption{Decision Rules}
    \label{alg:main}
    \algsize
    \begin{algorithmic}[1]
        \State \texttt{leadersPerRound} \Comment{A number between 1 and $q$}
        \State \texttt{waveLength} \Comment{Set to $2$ for \sysname}
        \Statex

        \Procedure{TryDecide}{$r_{committed}, r_{highest}$}
        \State $S \gets [ \; ]$ \Comment{Holds decisions}
        \For{$r \gets r_{highest}$ \textbf{down to} $r_{committed} + 1$}
        \For{$l \gets \texttt{leadersPerRound} - 1$ \textbf{down to} $0$}
        \State $i \gets r \; \bmod$ \texttt{waveLength}
        \State $D \gets$ \Call{Decider}{$i, l$}
        \State $w \gets D.$\Call{WaveNumber}{$r$}
        \State $s \gets D.$\Call{TryDirectDecide}{$w$}
        \If{$s = \bot$} $s \gets D.$\Call{TryIndirectDecide}{$w, S$}
        \EndIf
        \State $S \gets s \parallel S$
        \EndFor
        \EndFor
        \State \Return $S$
        \EndProcedure
        \Statex

        \Procedure{ExtendCommitSeq}{$r_{committed}, r_{highest}$}
        \State $S \gets$ \Call{TryDecide}{$r_{committed}, r_{highest}$}
        \State $S_{commit} \gets [ \; ]$ \Comment{Holds committed blocks}
        \For{$s \in S$}
        \If{$s = \bot$} \textbf{break}
        \EndIf
        \If{$s = \texttt{Commit}(b_{leader})$} $S_{commit} \gets S_{commit} \parallel b_{leader}$
        \EndIf
        \EndFor
        \State \Return \Call{LinearizeSubDags}{$S_{commit}$} \Comment{Same as DAG-Rider~\cite{dag-rider}}
        \EndProcedure

    \end{algorithmic}
\end{algorithm}

\begin{algorithm}[t!]
    \caption{Decider Instance and Helpers}
    \algsize
    \label{alg:decider}
    \begin{algorithmic}[1]

        \State \texttt{waveOffset} $= i$ \Comment{The first parameter of the Decider (i)}
        \State \texttt{leaderOffset} $= l$ \Comment{The second parameter of the Decider (l)}
        \State \texttt{waveLength} \Comment{Set to $2$ for \sysname}
        \State \texttt{replicas} \Comment{The set of replicas}
        \Statex

        \Procedure{WaveNumber}{$r$}
        \State \Return $(r - \texttt{waveOffset}) / \texttt{waveLength}$
        \EndProcedure
        \Statex

        \Procedure{ProposeRound}{$w$}
        \State \Return $(w * \texttt{waveLength}) + \texttt{waveOffset}$
        \EndProcedure
        \Statex

        \Procedure{DecisionRound}{$w$}
        \State \Return \Call{ProposeRound}{$w$}$ + (\texttt{waveLength} - 1)$
        \EndProcedure
        \Statex

        \Procedure{StronglyCertifiedLeader}{$w, b_{leader}$}
        \State $B_{decision} \gets$ \Call{GetDecisionBlocks}{$w$}
        \State \Return $|\{ b'.author : b' \in B_{decision} \land $ \Call{Support}{$b_{leader}, b'$}$ \}| \geq q$ \Comment{$q = n-f-c$; count authors, as replicas may equivocate}
        \EndProcedure
        \Statex

        \Procedure{SkippedLeader}{$w, b_{leader}$}
        \State $B_{decision} \gets$ \Call{GetDecisionBlocks}{$w$}
        \State \Return $|\{ b'.author : b' \in B_{decision} \land \neg$\Call{Support}{$b_{leader}, b'$}$ \}| \geq q$ \Comment{$q = n-f-c$}
        \EndProcedure
        \Statex

        \Procedure{TryDirectDecide}{$w$}
        \State $B_{leader} \gets$ \Call{GetLeaderBlocks}{$w$, \texttt{leaderOffset}}
        \For{$b_{leader} \in B_{leader}$}
        \If{\Call{SkippedLeader}{$w, b_{leader}$}} \Return \texttt{Skip}
        \EndIf
        \If{\Call{StronglyCertifiedLeader}{$w, b_{leader}$}} \Return \texttt{Commit}$(b_{leader})$
        \EndIf
        \EndFor
        \State \Return $\bot$
        \EndProcedure
        \Statex

        \Procedure{WeaklyCertifiedLeader}{$b_{anchor}, b_{leader}$}
        \State $w \gets$ \Call{WaveNumber}{$b_{leader}.round$}
        \State $B_{decision} \gets$ \Call{GetDecisionBlocks}{$w$}
        \State \Return $|\{ b.author : b \in B_{decision} \land $ \Call{Support}{$b_{leader}, b$} $\land$ \Call{Link}{$b, b_{anchor}$}$ \}| \geq k$ \Comment{$k = 2f+c+1$; votes are supports, anchor certification is reachability}
        \EndProcedure
        \Statex

        \Procedure{TryIndirectDecide}{$w, S$}
        \State $r_{decision} \gets $\Call{DecisionRound}{$w$}
        \State $s_{anchor} \gets$ first $s \in S$ s.t. $r_{decision} < s.round \land s \neq$ \texttt{Skip}
        \If{$s_{anchor} = \texttt{Commit}(b_{anchor})$}
        \State $B_{leader} \gets$ \Call{GetLeaderBlocks}{$w$, \texttt{leaderOffset}}
        \If{$\exists \; b_{leader} \in B_{leader}$ s.t. \Call{WeaklyCertifiedLeader}{$b_{anchor}, b_{leader}$}} \Return \texttt{Commit}$(b_{leader})$
        \Else \; \Return \texttt{Skip}
        \EndIf
        \EndIf
        \State \Return $\bot$
        \EndProcedure
        \Statex

        \Procedure{GetDecisionBlocks}{$w$}
        \State $r_{decision} \gets $\Call{DecisionRound}{$w$}
        \State \Return $DAG[r_{decision}]$
        \EndProcedure
        \Statex

        \Procedure{GetLeaderBlocks}{$w, rank$} \Comment{Replicas may equivocate}
        \State $r_{propose} \gets \Call{ProposeRound}{w}$
        \State $s \gets r_{propose}$
        \State $leader \gets \texttt{replicas}[(s + rank) \bmod |\texttt{replicas}|]$
        \State \Return $\{ b \in DAG[r_{propose}] : b.author = leader\}$
        \EndProcedure
        \Statex

        \Procedure{Link}{$b_{old}, b_{new}$}
        \State \Return $\exists$ a sequence of $m \in \mathbb{N}$ blocks $b_1, \ldots, b_m$ s.t.
        \Statex \hspace{2em} $b_1 = b_{old} \land b_m = b_{new} \land \forall j \in [2, m] : b_j \in \bigcup_{r \geq 1} DAG[r] \land b_{j-1} \in b_j.parents$
        \EndProcedure
        \Statex

        \Procedure{Support}{$b_{old}, b_{new}$}
        \State \Return $b_{old}$ is the first block by $b_{old}.author$ in round $b_{old}.round$ encountered
        \Statex \hspace{2em} in the deterministic depth-first traversal of $b_{new}$'s causal history \Comment{Mysticeti support~\cite{mysticeti}; selects at most one block per (author, round), even if the history contains equivocating copies}
        \EndProcedure

    \end{algorithmic}
\end{algorithm}

\section{Implementation}
\label{sec:implementation}

We implement a networked, multi-core \orcdag replica in Rust by forking the Mysticeti codebase~\cite{mysticeti-code,sui-lutris}. Our implementation leverages \texttt{tokio}~\cite{tokio} for asynchronous networking, utilizing raw TCP sockets for communication without relying on any RPC frameworks. For cryptographic operations, we rely on \texttt{ed25519-consensus}~\cite{ed25519-consensus,eddsa-bls-perf} for asymmetric cryptography and \texttt{blake2}~\cite{rustcrypto-hashes} for cryptographic hashing. To ensure data persistence and crash recovery, we employ a Write-Ahead Log (WAL). The WAL optimizes I/O operations through vectored writes~\cite{writev} and efficient memory-mapped file usage with the \texttt{minibytes}~\cite{minibytes} crate, minimizing data copying and serialization.

Both \orcdag and the Mysticeti baseline embed an identical checkpoint engine to support the Resilient Path (\cref{sec:protocol}). After executing a committed sub-DAG, where execution consists of deterministic parallel writes over the state, each replica derives a cryptographic commitment to the resulting state and forms a checkpoint attestation. The engine samples one checkpoint per committed sub-DAG, yielding ${\sim}25$ checkpoints per second (roughly one checkpoint per $40$\,ms of committed history). Following the practice of production blockchains~\cite{sui-code}, the engine introduces no custom checkpoint messages, thereby avoiding dedicated \textsc{ChkProp} and \textsc{ChkWitness} (\cref{alg:protocol}) network paths. Instead, it reuses consensus itself as the broadcast layer, submitting checkpoint attestations as small, specialized (\emph{system}) transactions; a checkpoint is certified once a quorum of matching attestations is ordered.

The key correctness point of this design is that ordering the attestation quorum implicitly yields the witness guarantee of the \finalityqc: once the quorum of matching attestations (forming a \checkpointqc) is ordered, every correct replica observes and stores the certified checkpoint as part of the total order, ensuring that the checkpoint is recorded without any dedicated witness round. For \orcdag, this mechanism realizes Resilient-Path finality; for the Mysticeti baseline ($n=3f+1$), the same engine provides durability only, since the fork bound at $3f+1$ leaves no alive-but-corrupt (AbC) equivocation slack. The cost of this reuse is latency: checkpoint certification completes one full ordering cycle after the transaction itself is ordered (introducing one extra message delay beyond the protocol-level bound). While dedicated checkpoint messages would save that delay, they would introduce new message types with distinct dissemination, synchronization, and recovery paths, which adds significant engineering complexity; we deliberately keep our engine production-faithful.

In addition to regular unit tests, we inherit and use two supplementary testing utilities from the Mysticeti codebase. First, a simulation layer replicates the functionality of the \texttt{tokio} runtime and TCP networking; the simulated network reproduces realistic WAN latencies, while the \texttt{tokio} runtime simulator employs a discrete-event simulation approach to model the passage of time. Second, a command-line utility (called the \emph{orchestrator})~\cite{narwhal,bullshark} deploys real-world clusters of \orcdag replicas on machines distributed across the globe.

We open-source our \orcdag implementation, along with its simulator and orchestration tools, to ensure reproducibility of our results\footnote{\consensuscode}\footnote{\validatorcode}.

\section{Testbed Details} \label{sec:experimental-setup}

This appendix complements \Cref{sec:evaluation} by detailing the testbed and experimental setup used to evaluate \orcdag.

We deploy all protocols on AWS, using \texttt{m5d.8xlarge} instances across $6$ different AWS regions: Northern Virginia (us-east-1), Ohio (us-east-2), Frankfurt (eu-central-1), London (eu-west-2), Paris (eu-west-3), and Tokyo (ap-northeast-1). Replicas are distributed across those regions as equally as possible. Each machine provides $10$\,Gbps of bandwidth, $32$ virtual CPUs (16 physical cores) on a $3.1$\,GHz Intel Xeon Skylake 8175M, $128$\,GB memory, and runs Linux Ubuntu server $24.04$.

We instantiate several geo-distributed benchmark clients within each replica submitting transactions in an open-loop model at a fixed rate. We experimentally increase the load of transactions sent to the systems, and record the throughput and latency of commits. As a result, all plots in \Cref{sec:evaluation} illustrate the steady-state latency of all systems under low load, as well as the maximal throughput they can provide after which latency grows quickly. Transactions in the benchmarks are arbitrary and contain $512$ bytes. We configure \orcdag and Mysticeti with $2$ leaders per round, and all protocols use a leader timeout of $1$ second.

All large-committee failure-free results are obtained from a single measurement campaign, using uniform committees of $n{=}51$ for every \orcdag configuration and $n{=}50$ for Mysticeti. For per-run statistics, we restrict each run to its throughput plateau (scrapes sustaining at least $90\%$ of the peak rate) and report the steady-state floor (the minimum p50 and p90 over this plateau) to ensure robustness to warm-up and cool down; throughput is reported as the median sustained rate. Each run lasts ${\sim}300$\,s, with metrics sampled every $15$\,seconds.

\section{Crash-Only Deployment} \label{sec:crash-only}

When configured with $f=0$ to tolerate only crashes and no Byzantine faults, \sysname reduces to an optimal crash-fault-tolerant (CFT) protocol~\cite{Paxos}. \Cref{fig:crash-only} reports the performance of \orcdag configured with $f=0$ and $c=1$ in a minimal $n=3$ deployment, where the three replicas run in three distinct, nearby European regions---Frankfurt (eu-central-1), Ireland (eu-west-1), and London (eu-west-2)---to model a regional CFT deployment. The plot shows two loads, $1$k and $10$k\,tx/s. Each box spans the mean $\pm$ $1$ standard deviation, and the whiskers denote $\pm$ $2$ standard deviations. In this lightly-loaded regime the commit latency is load-independent and essentially flat: \orcdag commits with a mean latency of about $22$\,ms at $1$k\,tx/s and about $23$\,ms at $10$k\,tx/s, with a standard deviation of about $7$\,ms. This is roughly two orders of magnitude below the hundreds of milliseconds observed in the geo-distributed WAN runs of the main evaluation (\Cref{sec:evaluation}), highlighting the best-case latency of a regional crash-tolerant deployment.

\begin{figure}[t]
    \centering
    \includegraphics[width=\linewidth]{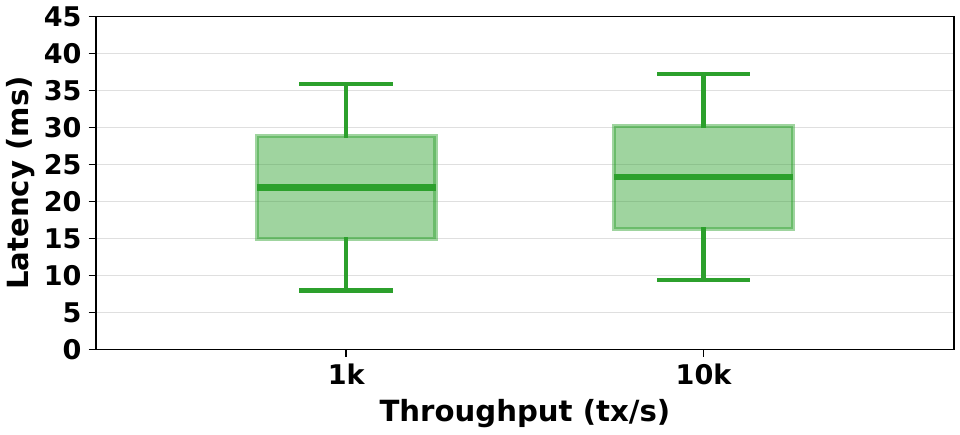}
    \caption{\footnotesize EU (multi-region) crash-only \orcdag{} run; mean $\pm$ stdev.}
    \label{fig:crash-only}
\end{figure}

\end{document}